\documentclass[pra,aps,twocolumn,showpacs,superscriptaddress,nofootinbib,floatfix]{revtex4-1}
\usepackage[utf8]{inputenc}
\usepackage[colorlinks=true,citecolor=blue,linkcolor=magenta]{hyperref}
\usepackage{graphicx}
\usepackage{caption}
\usepackage{subcaption}
\usepackage{dcolumn}
\usepackage{bm}
\usepackage[bbgreekl]{mathbbol}
\usepackage{amsmath,dsfont,mathtools}
\usepackage[normalem]{ulem}
\usepackage{color}
\usepackage[dvipsnames]{xcolor}
\usepackage{stmaryrd}

\usepackage{mathtools, stmaryrd}
\usepackage{amsthm}
\usepackage{comment}
\usepackage{tabularx}
\usepackage{algorithm}
\usepackage{algpseudocode}

\usepackage{booktabs}

\newcommand{\Sum}{\displaystyle\sum}

\newcommand{\Int}{\displaystyle\int}

\def\braket#1{\mathinner{\langle{#1}\rangle}}

\newcommand{\alphab}{\boldsymbol{\alpha}}
\newcommand{\norm}[1]{\left\lVert#1\right\rVert}
\newcommand{\w}{\boldsymbol{\mathrm{w}}}
\newcommand{\R}{\mathbb{R}}

\newcommand{\X}{\mathcal{X}}
\newcommand{\Y}{\mathcal{Y}}
\newcommand{\PP}{\mathbb{P}}
\newcommand{\Phib}{\boldsymbol{\Phi}}
\DeclareMathOperator*{\argmin}{\arg\!\min}

\newtheorem{theorem}{Theorem}

\begin{document}

%\title{Can Quantum Machine Learning Resist to Classical Fourier Sampling?}
%\title{Classical methods with the same expressivity than Variational quantum circuit for Machine Learning}
%\title{Variational Quantum Machine Learning can be Classically Reproduced with Random Fourier Features}
%title{Variational Quantum Machine Learning can be Classically Approximated with Random Fourier Features}
\title{Classically Approximating Variational Quantum Machine Learning with Random Fourier Features}
%\title{Will Classical Fourier Sampling Make Variational Quantum Machine Learning Useless?}
\author{Jonas Landman}
\affiliation{School of Informatics, University of Edinburgh, 10 Crichton Street, Edinburgh, United Kingdom}
\affiliation{QC Ware, Palo Alto, USA and Paris, France}

\author{Slimane Thabet}
\affiliation{Laboratoire d’Informatique de Paris 6, CNRS, Sorbonne Université, 4 Place Jussieu, 75005 Paris, France}
\affiliation{PASQAL SAS, 7 avenue Léonard de Vinci, 91300 Massy, France}

\author{Constantin Dalyac}
\affiliation{Laboratoire d’Informatique de Paris 6, CNRS, Sorbonne Université, 4 Place Jussieu, 75005 Paris, France}
\affiliation{PASQAL SAS, 7 avenue Léonard de Vinci, 91300 Massy, France}

\author{Hela Mhiri}
\affiliation{Laboratoire d’Informatique de Paris 6, CNRS, Sorbonne Université, 4 Place Jussieu, 75005 Paris, France}
\affiliation{ENSTA Paris, Institut polytechnique de Paris}

\author{Elham Kashefi}
\affiliation{School of Informatics, University of Edinburgh, 10 Crichton Street, Edinburgh, United Kingdom}
\affiliation{Laboratoire d’Informatique de Paris 6, CNRS, Sorbonne Université, 4 Place Jussieu, 75005 Paris, France}

\date{\today}
\begin{abstract}
    Many applications of quantum computing in the near term rely on variational quantum circuits (VQCs). They have been showcased as a promising model for reaching a quantum advantage in machine learning with current noisy intermediate scale quantum computers (NISQ). It is often believed that the power of VQCs relies on their exponentially large feature space, and extensive works have explored the expressiveness and trainability of VQCs in that regard. In our work, we propose a classical sampling method that may closely approximate a VQC with Hamiltonian encoding, given only the description of its architecture. It uses the seminal proposal of Random Fourier Features (RFF) and the fact that VQCs can be seen as large Fourier series. 
    We provide general theoretical bounds for classically approximating models built from exponentially large quantum feature space by sampling a few frequencies to build an equivalent low dimensional kernel, and we show experimentally that this approximation is efficient for several encoding strategies.
    %We show theoretically and experimentally that models built from exponentially large quantum feature space can be classically reproduced by sampling a few frequencies to build an equivalent low dimensional kernel. 
    Precisely, we show that the number of required samples grows favorably with the size of the quantum spectrum. This tool therefore questions the hope for quantum advantage from VQCs in many cases, but conversely helps to narrow the conditions for their potential success. We expect VQCs with various and complex encoding Hamiltonians, or with large input dimension, to become more robust to classical approximations.
\end{abstract}
\maketitle

\section*{Introduction}

\begin{figure*}[t!]
    \centering
    \includegraphics[width=\textwidth]{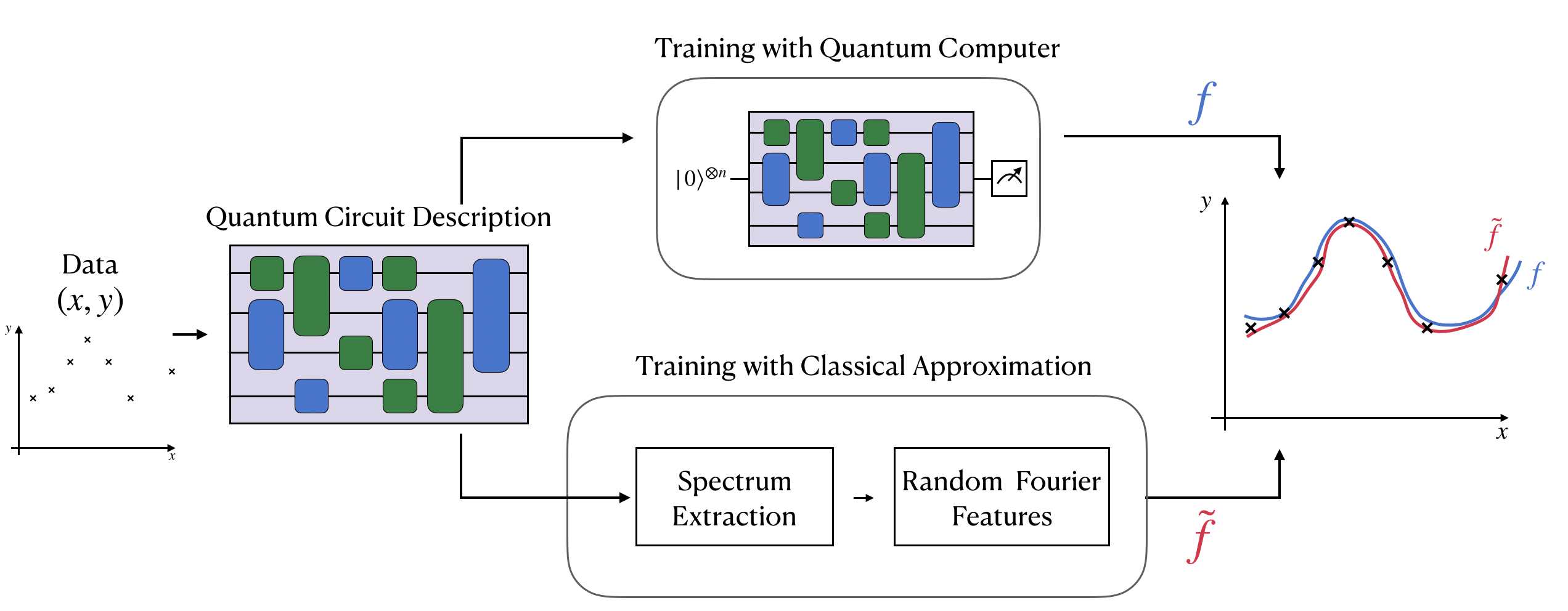}
    \caption{\textbf{Random Fourier Features as a classical approximator of quantum models.} Instead of training a Variational Quantum Circuit by using a quantum computer, we propose to train a classical kernel built by sampling a few frequencies of the quantum model. These frequencies can be derived from the quantum circuit architecture, in particular from the encoding gates. Using Random Fourier Features, one can build a classical model which performs as good as the quantum model with a bounded error and a number of samples that grows nicely.}
    \label{fig:Fig_1}
\end{figure*}

Many applications of quantum computing in the near term rely on variational quantum circuits (VQCs) \cite{bharti2021noisy, cerezo2021variational}, and in particular for solving machine learning (ML) tasks. VQCs are trained using classical optimization of their gates' parameters, a method borrowed from classical neural networks. Many early works have shown promising results, both empirically and in theory \cite{abbas2021power, cong2019quantum, huang2021power}. However, whether these variational methods can provide a quantum advantage in the general case with a scaling number of qubits has not been proven. 

What would make VQCs advantageous compared to classical algorithms for machine learning on classical data? The most common and intuitive answer is the formation of a \emph{large feature space}, due to the projection of data points in an exponentially large Hilbert space. Understanding what is happening in this Hilbert space \cite{Schuld2021}, and most importantly, knowing how we can exploit its size is an important question for this field. When it comes to trainability, it has been shown that these exponential spaces are in fact drawbacks \cite{mcclean2018barren}. Regarding expressivity, it is crucial to understand what kind of functions VQCs can learn better than a classical algorithm.

One of the biggest challenges we face while trying to answer these questions is to find the fair comparison between VQCs and classical ML algorithms. Recent works \cite{Schuld2021, schuld2021supervised, schreiber2022classical} have shown equivalence between VQCs and both Fourier series and kernel methods. In fact, VQCs are extremely powerful Fourier series: the functions they can learn are predetermined by a set of frequencies which can become numerous with the dimension of the input, and the number and complexity of the quantum gates used. 

In this work, we adapt Random Fourier Features (RFF) \cite{Rahimi2009}, a successful classical sampling algorithm aimed at efficiently approximating some large classical kernel methods. We designed three different RFF strategies to approximate VQCs. For each one, we analyze its efficiency in reproducing results obtained by a VQC. To do so, we have studied in details the expressivity of VQCs, understood where their power could come from, and compared it each time with RFF.  

Our method consists in analyzing the encoding gates of the VQC to extract the final frequencies of its model and sample from them. Notably, if the VQCs possesses simple encoding gates such as Pauli gates% (e.g. $R_Z$)
, we show that the large quantum feature space is not fully exploited, making RFF even more efficient. If the number of frequencies in the VQC grows exponentially, the number of samples required by RFF grows only linearly. Finally, we have empirically compared VQCs and RFF on real and artificial use cases. On these, RFF were able to match the VQC's answer, and sometimes outperform it. 
Some hope resides for VQCs with non-diagonalizable encoding Hamiltonians, preventing the use of usual RFF. However, even in this case, we provide an alternative RFF strategy, along with theoretical approximation bounds. 

Our conclusion is that, despite the potentially exponential number of frequencies in the functions that a VQC can create, Random Fourier Features can be used to approximate the same resulting function on a given dataset, in many cases. This restrains the regime in which VQCs can offer a true advantage over classical in learning tasks over classical data, or invites us to change the way VQCs are defined.

The manuscript is structured as follows. In Section \ref{sec:VQC_preliminaries}, we recall the properties of variational quantum circuits, and the fact that they can be expressed as Fourier series by detailing the structure of their spectrum. In Section \ref{sec:classical_RFF}, we introduce the classical method of Random Fourier Features. In Section \ref{sec:method_and_sampling}, we finally present three ways of applying RFF to approximate VQC's models. These theoretical proposals are tested over real and artificial datasets in Section \ref{sec:exp_numerics}.

\section{Spectral Properties of Variational Quantum Circuits}
\label{sec:VQC_preliminaries}

\subsection{Definitions}

We consider a standard ML task where a function $f$, named \emph{model}, must be optimized to map data points to their target values. The data used for the training is made of $M$ points $x=(x_1, \dots, x_d)$ in $\X=\R^d$ along with their target values $y$ in $\Y=\R$. We define a \textit{quantum model} as the family of parametrized functions $f: (\X, \Theta) \longrightarrow \Y$, such that
\begin{equation}
\label{eq:quantum_model}
    f(x; \theta) = \langle 0|U(x; \theta)^\dagger O U(x; \theta) |0\rangle
\end{equation}
    
where $U(x; \theta)$ is a unitary that represents the parametrized quantum circuits, $\theta$ represents the trainable parameters from a space $\Theta$, and $O$ is an observable. We can always describe the parametrized quantum circuit as a series of two types of gates. The first are called \emph{encoding} gates as they only depend on input data values, whereas the \emph{trainable} gates depend on internal parameters that are optimized during training. A typical instance of a Variational Quantum Circuit (VQC) is illustrated in Fig.\ref{fig:VQC_explained}. In the recent literature, these gates are often grouped as \emph{layers}, which is not mandatory, since any circuit can be sliced into alternating sequences of encoding and training blocks (even if containing a single gate). 

Any quantum unitary implements the evolution of a quantum system under a Hamiltonian. Thus, we choose to write the $\ell^{th}$ encoding gates as $\exp(-ix_i H_{\ell})$, where $x_i$ is one of the $d$ components of $x$, and $H_{\ell}$ is a Hamiltonian matrix of size $2^p$ if $p$ is the number of qubits this gate acts on. We will note $L$ the number of encoding gates for each dimension of $x$ (the same for each dimension, for notation simplicity).

%The global unitary can therefore be written as
%\jonas{\begin{equation}
%     U(x; \theta) = \Prod_{k=1}^{L} \Prod_{p=1}^{d} %\big[
%     U_{kp}(\theta_{kp})\exp(-ix_pH_{kp}) 
%     %\big]
%     %U_0(\theta_0)
%\end{equation}}
%% \\
%\jonas{(check the above formula. not clear. $L$ should be the number of encoding gates per dimension (assuming its the same number for each dim). what is $d$ here?)} where $U_{kp}(\theta_{kp})$ is a trainable unitary parameterized by the vector $\theta_{kp}$ and $H_{kp}$ is an encoding Hamiltonian. 

% \proposal{other formulation}
% \proposal{--------------------------------------------}

%\begin{equation}
%\label{eq:U_x_theta}
%    U(x;\theta)= \Prod_{k=1}^L U_k(\theta_k)S_k(x)
%\end{equation}

%where 
%\begin{equation}
%    S_k(x) = e^{-ix_1H_1^k} \otimes \cdots \otimes e^{-ix_d H_d^k}
%\end{equation}

%is the $k-$th data encoding layer with fixed Hamiltonians $H_i^k$ and $U_k(\theta_k)$ is a trainable unitary parametrized by the vector $\theta_k$.
~\\
\indent In this framework, the aim is to find the optimal mapping between data points and their target values. This is done by optimizing the parameters $\theta$ to find the best guess $f^*$ such that

\begin{gather}
    f^* = \argmin_\theta \frac{1}{M}\Sum_{i=1}^M l(f(x_i ; \theta), y_i)
\end{gather}
where $l$ is a cost function adapted to the task. For a standard regression tasks, we can choose $l(z, y) = |z-y|^2$ %$$\|z-y \|_2^2$. %jonas: it's not a vector.

\subsection{Quantum Models are Large Fourier Series}
\label{sec:quantum_models_are_fourier}
 We know since \cite{Schuld2021} that the family of quantum models defined in Eq.(\ref{eq:quantum_model}) can be rewritten as a Fourier series:

\begin{equation}
\label{eq:f_as_fourier_serie_complex}
    f(x;\theta) = 
\sum_{\omega\in\Omega} c_{\omega}e^{i\omega x}
\end{equation}

where the spectrum $\Omega$ of frequencies is determined by the ensemble of eigenvalues of the encoding Hamiltonians and the coefficients $c_\omega$ depend on the parametrized ansatz, as pictured in Fig.\ref{fig:VQC_explained}.

In order to familiarize the reader with the structure of the spectrum, we explicitly build $\Omega$ in the case of a one dimensional data input ($\X = \mathbb{R}$) and with a variational circuit containing only $L$ encoding gates. The accessible frequency spectrum $\Omega$ is the ensemble of all the differences between all possible sums of the eigenvalues of the encoding gates as shown in Fig.\ref{fig:tree_sampling}. 

We note $\lambda_\ell^{k}$ the $k^{th}$ eigenvalue of the $\ell^{th}$ encoding Hamiltonian $H_\ell$ having $d_{\ell}$ eigenvalues. We use the multi-index  $\boldsymbol{i} = (i_1, \dots, i_L)$ indicating which eigenvalue is taken from each encoding Hamiltonian. We define $\Lambda_{\boldsymbol{i}}$ as
\begin{equation}
    \Lambda_{\boldsymbol{i}} = \lambda_1^{i_1} + \dots + \lambda_L^{i_L}
\end{equation} 

Finally, we can express $\Omega$, the set of frequencies as in Eq.\ref{eq:Omega_definition_one_dimension}: %For simplicity, we wrote this formulation assuming all Hamiltonians have 2 eigenvalues ($\textbf{i}\in\{0,1\}^L$), but could be generalized to any number of eigenvalues per Hamiltonian:

\begin{equation}
\label{eq:Omega_definition_one_dimension}
\Omega = \bigg\{ \Lambda_{\boldsymbol{i}} - \Lambda_{\boldsymbol{j}}, \boldsymbol{i}, \boldsymbol{j} \in \prod_{\ell=1}^L [|1,d_{\ell}|] \bigg\},
\end{equation}

The simplest case is called Pauli encoding, when all encoding Hamiltonians are in fact Pauli matrices (e.g. encoding gates $R_Z(x) = e^{-i\frac{x}{2}\sigma_Z}$) as in  \cite{Schuld2021, Caro2021encodingdependent}. In this case, all the eigenvalues are $\lambda = \pm 1/2$, and therefore, the $\Lambda_{\boldsymbol{i}}$ are all the integers (or half-integers, if $L$ is odd) in $[-L/2, L/2]$. It follows that the set of distinct values in $\Omega$ is simply the set of integers in $\llbracket -L, L \rrbracket$. In this case, there are many redundant frequencies, due to the fact that all Pauli eigenvalues are the same. Namely, only $2L+1$ distinct frequencies among the $2^{2L}$ possible values of $\Lambda_{\boldsymbol{i}} - \Lambda_{\boldsymbol{j}}$. As shown in Fig.\ref{fig:tree_sampling}, more various eigenvalues would create more distinct frequencies in the end. In the rest of the paper, $\Omega$ will denote the set of distinct frequencies, without redundancy. 

In our experiments (Section \ref{sec:exp_numerics}), we  observe an important phenomenon: redundant frequencies are likely to have high coefficients. %(for both random and trained VQCs). 
Unique frequencies might often have in contrast small coefficients, reducing the potential expressivity of the VQC. We see that the redundancy might therefore play an important role in the expressivity of VQCs, and leave theoretical proof for future work. In fact, the impact of the frequencies redundancy on the quantum model has been recently highlighted in \cite{schuld2022generalization} where an analytical correspondence between a frequency redundancy and its Fourier coefficient has been proven for a simplified class of parameterized quantum models.%\footnote{The considered quantum circuit starts with random state initialization using a random d-dimensional unitary $U$ followed by one encoding gate which implements a d-dimensional Hamiltonian evolution and finally an observable $M$. In this model, only 1-dimensional input data is considered. }. More precisely, the frequencies distribution, fully determined by the encoding strategy, predetermines a fixed weight component in each frequency coefficient resulting in a strong prior that favors the corresponding highly redundant Fourier features. }

These arguments give some intuition on why one should use encoding gates from Hamiltonians with rich and various eigenvalues, by taking complex interactions over many qubits. A global Hamiltonian over $n$ qubits, hard to implement, could potentially have $2^n$ distinct eigenvalues, thus enlarging $\Omega$ and avoiding redundancy. Another approach from \cite{shin2022exponential} consists in adding scaling factors in the Pauli encoding gate to modify their eigenvalues and avoid redundancy. It results in an exponential number of integer frequencies, with respect to $L$, with many very high frequencies.
%\jonas{(On devrait essayer de fitter leur VQC avec des RFF pour $L=1,..,14$ car $3^L=10^7$ notre nb max de point en mémoire (divisé par 2, citère de shannon?) et ça serait une bonne expérience pour faire le lien entre $D$ et $|Omega|$ sans dépendance envers $d$!)}

\begin{figure}[h]
    \centering
    \includegraphics[width=0.5\textwidth]{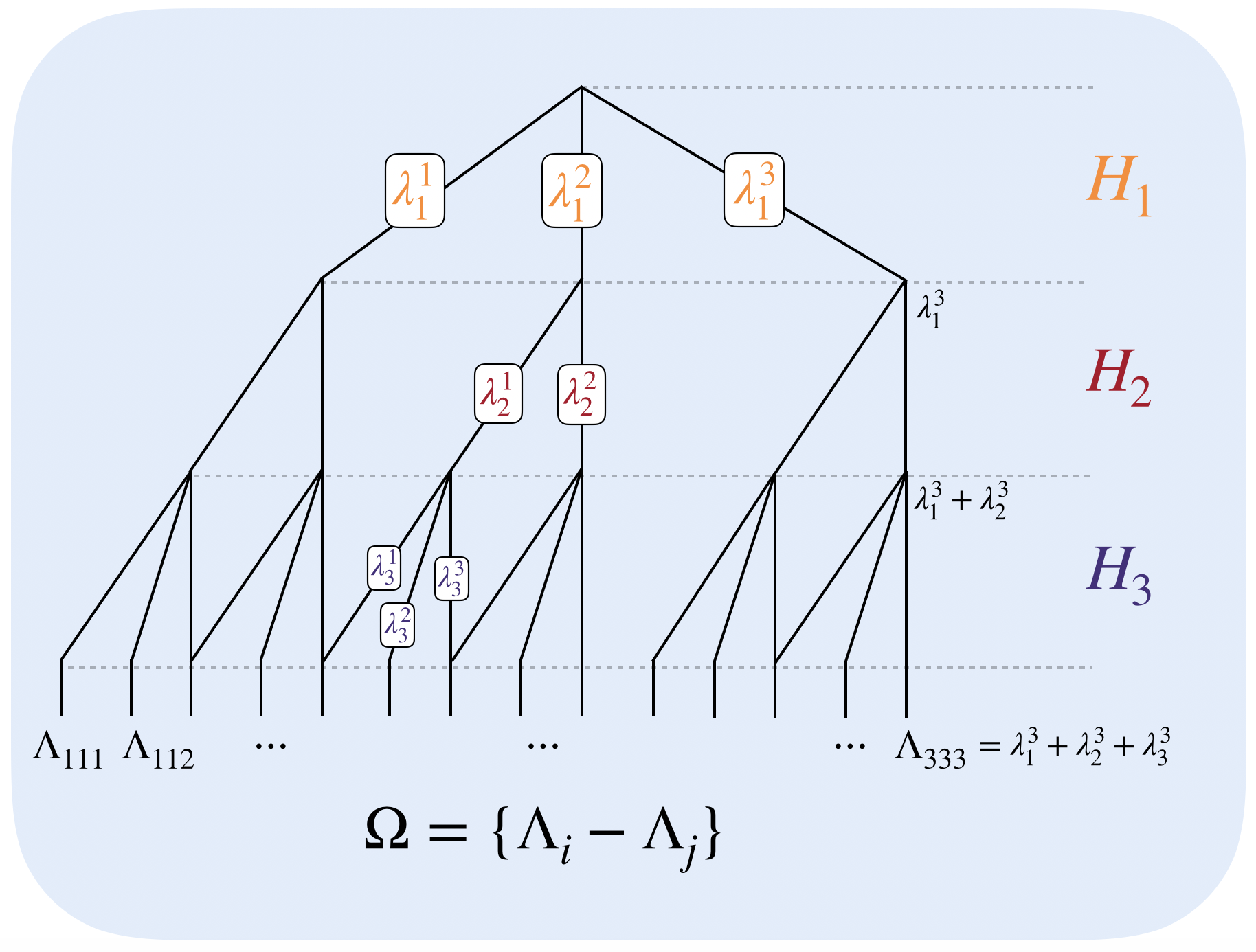}
    \caption{\textbf{From encoding Hamiltonians to Frequencies.} The frequencies composing the VQC model (on one dimensional input) come from all the combinations of eigenvalues from each encoding Hamiltonians. This can be seen as a tree, with $L=3$ Hamiltonians in this figure. We also see potential redundancy in the leaves.}
    \label{fig:tree_sampling}
\end{figure}

We can now generalize, if we now suppose that $\X = \R^d$, such that we encode a vector $x=(x_1, \dots, x_d)$ in our quantum model, then $\Omega$ becomes the following $d-$dimensional Cartesian product

\begin{equation}
\label{eq:Omega_general}
    \Omega = \Omega_1 \times \Omega_2 \times \dots  \times \Omega_d 
\end{equation}

where each $\Omega_\kappa$ is defined as in Eq.\ref{eq:Omega_definition_one_dimension} on its own set of Hamiltonians.

%\begin{gather}
%\label{eq:Omega_formula_dimension_d}
%    \Omega_\kappa = \bigg\{ \Lambda^\kappa_{\boldsymbol{i}} - \Lambda^\kappa_{\boldsymbol{j}}, \boldsymbol{i}, \boldsymbol{j} \in [p^\kappa]^L\bigg\}
%\end{gather}

%where $p^\kappa$ is the maximal dimension associated to the encoding Hamiltonians of $x_\kappa$ 
In this context, note that the frequencies $\omega$ are now vectors in $\R^d$ and there are $d$ different trees to build $\Omega$ (see Fig. \ref{fig:tree_sampling}). Note that for notation simplicity, we assumed that $L$ gates were applied on each input's component, but it can be generalized to any number of gates per dimension.

We therefore see that the size of the spectrum $|\Omega|$ can potentially grow exponentially with the number of encoding gates and the dimension of the input data. For instance, if we consider a $d$-dimensional vector $x$ and $L$ Pauli-encoding gates for each dimension in such a way that there are $Ld$ encoding gates in the VQC. According to equation (\ref{eq:Omega_general}),
the size of the spectrum $\Omega$ would scale as $O(L^d)$, which becomes quickly intractable as $d$ increases. As an example, the spectrum associated to a VQC with $L=20$ encoding gates and $d=16$ would require more than one hundred times the world's storage data capacity available in 2007 to be stored \cite{doi:10.1126/science.1200970}.

However, can such large Fourier series be approximated with classical methods? It would consist in building a classical approximator $\tilde{f}$ as %We therefore wonder if it is possible to build a classical approximator
\begin{gather}
    \tilde{f}(x) = \Sum_{\omega \in \tilde{\Omega}} \tilde{c}_\omega e^{i\omega x}
\end{gather}
such that $\tilde{\Omega}$ is of tractable size and the two solution are close, written as
$\norm {f(x) - \tilde{f}(x)} \leq \varepsilon$, using a given error measure (infinite or $\ell_2$ norms, for instance).

In the cases where the construction of such a model is possible, it would imply that although classically simulating the encoding VQC might not be possible, the quantum model that emerges from it can be efficiently approximated in a classical way.

\subsection{Quantum models are shift-invariant kernel methods}
\label{sec:quantum_models_are_shifinvariantkernel}

As the quantum model is a real-valued function, it follows that $\omega \in \Omega$ implies $-\omega \in \Omega$ and $c_\omega = c_{-\omega}^*$. We express the Fourier series of the quantum model as a sum of trigonometric functions by defining for every $\omega \in \Omega$:

\begin{gather}
    a_\omega := c_\omega + c_{-\omega} \in \mathbb{R} \\
    b_\omega := \frac{1}{i}(c_\omega - c_{-\omega}) \in \mathbb{R}
\end{gather}

such that 

\begin{equation}
\label{eq:VQC_function_analytic}
\begin{split}
     f(x;\theta) &= 
\sum_{\omega\in \Omega_+} c_{\omega}e^{i\omega x}+ c_{-\omega}e^{-i\omega x}\\ 
&= \sum_{\omega\in \Omega_+} a_{\omega}\cos(\omega x) + b_{\omega}\sin(\omega x)
\end{split}
\end{equation}

where $\Omega_+$ contains only half of the frequencies from $\Omega$. 
Considering only Pauli gates, if $d = 1$, we simply have $\Omega = \llbracket -L,L \rrbracket$ and $\Omega_+ = \llbracket 0,L \rrbracket$.
In dimension $d$,  we have $\Omega = \llbracket -L,L \rrbracket^d$ and $\Omega_+$ is built by keeping half of the frequencies (after removing those of opposite sign), plus the null vector. In the end, we have 
\begin{equation}
\label{eq:size_of_Omega_plus}
|\Omega_+| = \frac{(2L+1)^d-1}{2} +1
\end{equation}
%With a more general encoding scheme, the formula stays the same, but $L$ would represent the (same) number of distinct positive frequencies per dimension. 
With a more general encoding scheme, if there is a different number of distinct positive frequencies per dimension, the formula is different but is built similarly.

In the following parts, we will focus solely on $\Omega_+$ and conveniently drop the $+$ subscript.

%\jonas{(in 1d, yes, but not in $d$? give details: puisque $\omega$ et $-\omega$ représentent la même fréquence, vu comme des vecteurs la vraie formule est $(2L)^d/2$, avec $L$ le nb de gates par dimension pour des Paulis, puis se generalize pour des non Pauli si on dit que $L$ devient alors le nombre de fréquences positives disctinctes par dimensions. et d'ailleurs dans le cas Pauli il faut écrire $\Omega=[|0,L|]^d$ et pas $\Omega=[|1,L|]^d$, qu'on peut écrire $[|L|]^d$ maybe.)}. 

\begin{figure}[h]
    \centering
    \includegraphics[width=0.35\textwidth]{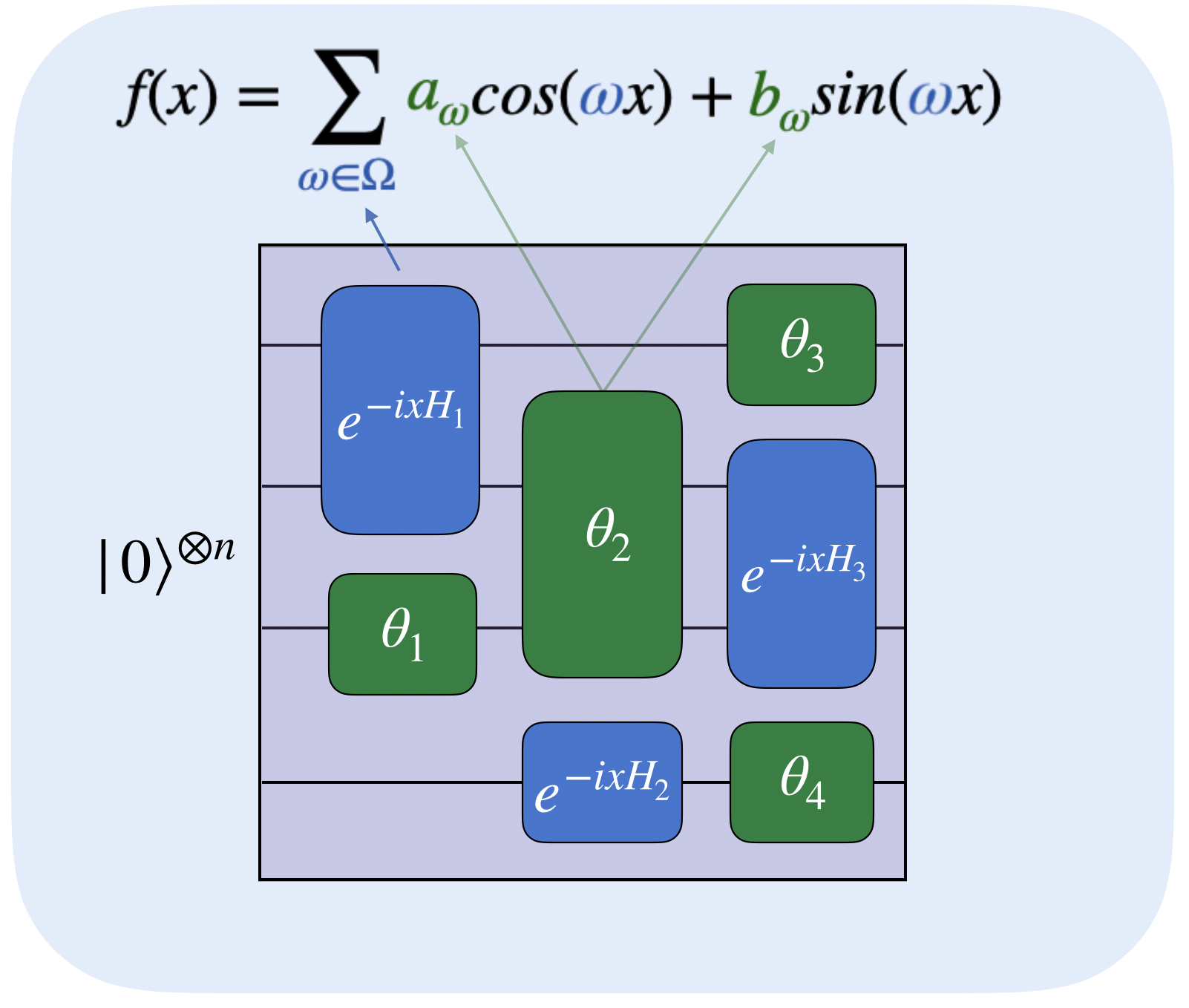}
    \caption{\textbf{Variational Quantum Circuits give rise to Fourier Series.} In a quantum machine learning task, classical data is encoded in a subset of variational gates of a quantum circuit (green), while blue gates are trainable.}
    \label{fig:VQC_explained}
\end{figure}

%Given this formulation of the quantum model 
We can also define the feature map of the quantum model \cite{schuld2021supervised} as 
\begin{equation}
f(x;\theta) = \braket{\psi(x;\theta)|O|\psi(x;\theta)} = \w(\theta)^T\phi(x)
\end{equation}

where $\phi(x)$ is the \emph{feature vector}, the mapping of the initial input into a larger \emph{feature space}, where the new distribution of the data is supposed to make the classification (or regression) solvable with only a linear model. This linear model is in fact the inner product between $\phi(x)$ and a trainable \emph{weight vector} $\w$. In the case of VQCs, we can explicitly express them as:

\begin{equation}
\label{eq:phi_x_definition}
    \phi(x) = \frac{1}{\sqrt{|\Omega|}}
\begin{bmatrix}
  cos(\omega^T x)\\
  sin(\omega^T x)\\
  \vdots\\
\end{bmatrix}_{\omega \in \Omega}, \quad  
\w(\theta) = 
\begin{bmatrix}
  a_\omega\\
  b_\omega\\
  \vdots\\
\end{bmatrix}_{\omega \in \Omega}
\end{equation}

If the spectrum $\Omega$ is known and accessible, one can fit the quantum model by retrieving the coefficients $a_\omega, b_\omega$ associated to each frequency $\omega$. This can be done by using general linear ridge regression techniques (see Appendix \ref{app:lrr}). Interestingly, there exists a dual formulation of the linear ridge regression that depends entirely on the kernel function associated to the model\,\cite{bishop2006pattern}. In our case, the related kernel function is defined by: 
\begin{equation}
\label{eq:kernel_VQC_definition}
\begin{split}
    k(x,x') &= \braket{\phi(x),\phi(x')} \\
&= \frac{1}{|\Omega|}\sum_{\omega \in \Omega}\cos(\omega x)\cos(\omega x') + \sin(\omega x)\sin(\omega x') \\
 &= \frac{1}{|\Omega|}\sum_{\omega \in \Omega}\cos(\omega (x-x'))
\end{split}
\end{equation}

which is a shift-invariant kernel, meaning that $k(x,x')=k(x-x')$. 

It is known that quantum models from VQCs are equivalent to kernel methods\,\cite{schuld2021supervised}, which means that it is equally possible to fit the quantum model by approximating the related kernel function. These kernels are high dimensional (since the frequencies in $\Omega$ can be numerous) which makes it hard to simulate classically in practice. But due to their shift-invariance, we propose to study their classical approximation using Random Fourier Features (RFF), a seminal method known to be powerful approximator of high-dimensional kernels\,\cite{Rahimi2009}.

% \subsection{Fitting the quantum model}

% If the spectrum $\Omega$ is known and accessible, one can fit the quantum model by finding the correct coefficients $a_\omega, b_\omega$ associated to each frequency $\omega$. Namely, one minimizes a regularized sum-of-squares error over $M$ training points to determine $\w^*$ such that

% \begin{gather}
%     \w^* = \argmin_{\w} \frac{1}{M}\Sum_{i=1}^M|\w^T\phi(x_i) - y_i|^2 + \frac{\lambda}{M}||\w||^2
% \end{gather}
% where  $\phi(x) = \begin{bmatrix}cos(\omega x) \\ sin(\omega x)  \end{bmatrix}_{\omega \in \Omega}$, and $\lambda$ is a regularisation parameter. With a dataset of $M$ points, this can be solved exactly using matrix inversion in $\mathcal{O}(M|\Omega|^2+|\Omega|^3)$ operations if $M\geq 2|\Omega|$. If the inequality is not fulfilled or if $|\Omega|^3$ is too big, one would use stochastic gradient descent instead of matrix inversion. \jonas{needs classical ML refs please}.

% Interestingly, there exists a dual formulation of this minimization formula that involves the kernel associated to the feature map $\phi$. 

\section{Random Fourier Features approximates high-dimensional kernels}
\label{sec:classical_RFF}
In this section, we explain the key results of the classical method called Random Fourier Features (RFF) \cite{Rahimi2009, li2019towards, sutherland2015error}. We will use this method to create several classical sampling algorithms for approximating VQCs.\\
%We first recall the kernel trick and how it can be used for regression. Then say that we can approximate it, or sample it. 

Let  $\mathcal{X} \subset \mathbb{R}^d$ be a compact domain and $k : \mathcal{X} \times \mathcal{X} \longrightarrow \mathbb{R}$ be a kernel function. We assume $k$ is shift invariant, meaning 
\begin{equation}
\label{eq:shift_invariant_kernel}
    k(x, y) = \overline{k}(x-y) = \overline{k}(\delta)
\end{equation}
where $\overline{k} : \mathcal{X} \longrightarrow \mathbb{R}$ is a single variable function, and we will note $\overline{k} = k$ to simplify the notation.

Bochner's theorem\,\cite{rudin2017fourier} insures that the Fourier transform of $k$ is a positive function and we can write

\begin{gather}
    k(\delta) = \Int_{\omega \in \mathcal{X}} p(\omega) e^{-i\omega^T\delta}d\omega
\end{gather}

If we assume $k$ is also normalized, then the Fourier transform $p(\omega)$ of $k$ can be assimilated to a probability distribution. With a dataset of $M$ points, fitting a Kernel Ridge Regression (KRR) model with the kernel $k$ necessitates $M^2$ operations to compute the kernel matrix and $\mathcal{O}(M^3)$ to invert it. This becomes impractical when $M$ reaches high value in modern big datasets.

The idea of the Random Fourier Feature method \cite{Rahimi2009} is to approximate the kernel $k$ by 
\begin{equation}
\label{eq:simeq_RFF}
    \tilde{k}(x, y) \simeq \tilde{\phi}(y)^T\tilde{\phi}(x)
\end{equation}
where $\tilde{\phi}(x) = \frac{1}{\sqrt{D}}
\begin{bmatrix}
cos(\omega_i^T x) \\ sin(\omega_i^T x) 
\end{bmatrix}_{i\in \llbracket 1, D \rrbracket}$. where the $\omega_i$s are $D$ frequencies sampled iid from the frequency distribution $p(\omega)$. Formally, it is a Monte-Carlo estimate of $k$. Note that $p(\omega)$ can be analytically found in some cases such as Gaussian or Cauchy kernel \cite{Rahimi2009}.

%\jonas{(donner des exemples. cite \cite{Rahimi2009}. Et si on a pas de formule analytique, c'est mort?)}

 Then instead of fitting a KRR for $k$, one will solve a Linear Ridge Regression (LRR) with $\phi$ (see details in Appendix \ref{app:lrr}). %\jonas{(@Slimane explain LRR, and explain RidgeLRR)}.
 The two problems are equivalent \cite{bishop2006pattern}, and the number of operations needed for the LRR is $\mathcal{O}(MD^2 + D^3)$. If D is much smaller than $N$, it is much cheaper than solving the KRR directly. Even if D is so big that the linear regression cannot be exactly solved, one can employ stochastic gradient descent or adaptive momentum optimizers such as Adam \cite{kingma2014adam}. 
 
 The output of the LRR or gradient descent is simply a weight vector $\tilde{\w}$ that is used to create the approximate function
 
 \begin{equation}
     \tilde{f} = \tilde{\w}^T\tilde{\phi}(x)
 \end{equation}

\section{RFF Methods for Approximating VQCs}
\label{sec:method_and_sampling}

In this section, we present in details our solutions to approximate a VQC using classical methods. The intuitive idea is to sample some frequencies from the VQC's frequency domain $\Omega$, and train a classical model from them, using the RFF methods from Section \ref{sec:classical_RFF}. We first introduce some related work (\ref{sec:related_work_jenseisert}), then present three different strategies to sample those frequencies to build classical models (\ref{sec:3sampling_strategies_RFF}). %and detail in which cases we expect our methods to perform well. 
We finally provide theoretical bounds on the number of samples required (\ref{sec:RFF_error_num_samples}) and present potential limitations to our methods (\ref{sec:limitations_QuantumRFF}), opening the way for VQCs with strong advantage over classical methods.

\subsection{Related Work}
\label{sec:related_work_jenseisert}

%Our procedure requires the description of the Variational Quantum Circuit. From the data-encoding gates, we implement several strategies to sample from the spectrum of the VQC.

A recent work \cite{schreiber2022classical} independently proposed a similar approach where classical surrogate methods approximate VQCs. The difference with this work is the necessity of having access to all $\Omega$, the totality of the frequencies of the VQC considered, without sampling from them. 

Indeed, if $\Omega$ is known, the coefficients $a_\omega$ and $b_\omega$ of the VQC function (see  Eq.\ref{eq:VQC_function_analytic}) can be easily fitted by solving the classical least square problem. Namely, one determines $\w^*$ such that
\begin{gather}\label{eq:least_square}
    \w^* = \argmin_{\w} \frac{1}{M}\Sum_{i=1}^M|\w^T\phi(x_i) - y_i|^2 + \frac{\lambda_0}{M}||\w||^2
\end{gather}
where  $\phi(x) = \begin{bmatrix}cos(\omega x) \\ sin(\omega x)  \end{bmatrix}_{\omega \in \Omega}$, and $\lambda_0$ is the regularisation parameter. As explained in the previous section, with a dataset of $M$ points, this can be solved exactly using matrix inversion in $\mathcal{O}(M|\Omega|^2+|\Omega|^3)$ operations if $M\geq 2|\Omega|$. If the inequality is not fulfilled or if $|\Omega|^3$ is too big, one would use stochastic gradient descent instead of matrix inversion.

However, this method assumes that $\Omega$ is known, and is not too large, which will usually be the case as we shown in Section \ref{sec:quantum_models_are_fourier}. One should also be able to enumerate all individual frequencies $\omega \in \Omega$. Moreover, as we will show the redundancy of some frequencies in $\Omega$ has a key importance, which is not captured by such a method.

For completeness, we note from the seminal work \cite{schuld2021supervised} that the author briefly mentions the idea of approximating kernels with RFF. Similarly, in a more recent work \cite{peters2022generalization}, the authors mention RFF as a sampling strategy on VQCs with shift invariant kernels, without further details. 

%\jonas{REMOVE, maybe not true.} Note that in \cite{schreiber2022classical}, the authors use this technique not by using $M$ pre-existing data points, but by taking $M$ measurements out of the VQC, in order to approximate it. In this work, our focus is more on the final models $f$ and $\tilde{f}$, respectively from the VQC and its approximator, by directly using the input dataset to train the approximator. Therefore, as shown in Fig.\ref{fig:Fig_1}, our method only needs the description of the quantum circuit, and not an actual quantum computer to takes samples.\\

\subsection{RFF Sampling strategies}
\label{sec:3sampling_strategies_RFF}

We now propose 3 types of sampling strategies of the spectrum $\Omega$. We will explain in which case these strategies are suited, according to the type of the encoding circuit, the dimension of the input vectors, the number of training points, and so on.

As shown in Eq.\ref{eq:phi_x_definition} and \ref{eq:kernel_VQC_definition}, the corresponding kernel $k$ of the VQC is built from frequencies $\omega \in \Omega$. In section \ref{sec:quantum_models_are_shifinvariantkernel}, we have shown that this kernel is shift-invariant, which ensures us the efficiency of RFF to approximate them. Applying the RFF method would consist in sampling from $\Omega$ and reconstructing a classical model which should approximate the VQC output function $f$.

%In this case for 1-dimensional input $x$ a classical computer would be able to reproduce the same model. Indeed, one would compute the feature vector $\phi(x) = \begin{bmatrix}cos(x) & sin(x) & \dots & \cos(L x) & \sin(L x) \end{bmatrix}^T$ which has a size linear in the number of encoding gates, and solve a least-square problem while performing a linear regression. This idea has been proposed in \cite{schreiber2022classical}.

%\jonas{ajouter un diagram / ajouter un pseudo code}

% We analyse in this section the warm-up case where the encoding layers are Pauli Z rotations in the light of the introduced theoretical tools and the resulting spectrum $\Omega$ consists of the integer-valued vectors $\llbracket0, L \rrbracket^d$.

\subsubsection{RFF with Distinct sampling}
This strategy describe the basic approach of using RFF for approximating VQCs. We assume the ability to sample from $\Omega$. It is straightforwardly following the method described in Section \ref{sec:classical_RFF}, and given in Algorithm \ref{alg:distinct_sampling}.

The benefit of this method, compared to the one presented in the previous Section \ref{sec:related_work_jenseisert}, is the ability to use far fewer frequencies than the actual size of $\Omega$. As shown in Section \ref{sec:RFF_error_num_samples}, one might require a number $D$ of samples scaling linearly with the dimension of the input $d$, and logarithmically with the size of $\Omega$.

\begin{algorithm}[H]
\caption{RFF with Distinct Sampling}
\label{alg:distinct_sampling}
\begin{algorithmic}[0]
\Require a VQC model $f$, and $M$ datapoints $\{x_j\}_{j\in[M]}$
\Ensure Approximate function $\tilde{f}$
\State Diagonalize the Hamiltonians of the VQC's encoding gates.
\State Use their eigenvalues to obtain all frequencies $\omega \in \Omega$, as in Eq.\ref{eq:Omega_definition_one_dimension}
\State Sample $D$ frequencies $(\omega_1,\cdots,\omega_D)$ from $\Omega$
\State Construct the approximated kernel $\tilde{k}(x,y)=\tilde{\phi}(y)^T\tilde{\phi}(x)$ with $\tilde{\phi}(x) = \frac{1}{\sqrt{D}}\begin{bmatrix}cos(\omega_i^T x) \\ sin(\omega_i^T x)  \end{bmatrix}_{i\in \llbracket 1, D \rrbracket}$
\State Solve the Linear Ridge Regression problem (Appendix \ref{app:lrr}), by matrix inversion or stochastic gradient descent, obtain a weight vector $\tilde{\w}$.
\State Obtain the approximated function $\tilde{f}(x) = \tilde{\w}^T\tilde{\phi(x)}$
\end{algorithmic}
\end{algorithm}

\subsubsection{RFF with Tree sampling}
\label{sec:tree_sampling}

The abovementioned method requires constructing explicitly $\Omega$, which can become exponentially large if the dimension $d$ of the datapoints is high (see Section \ref{sec:quantum_models_are_fourier}). The size of Omega is also increased if the encoding Hamiltonians have many eigenvalues (when the Hamiltonian is complex and acts on many qubits), or if many encoding gates are used. A large $\Omega$ is indeed the main interest of VQCs in the first place, promising a large expressivity. 

In some cases, and in particular for a VQC using many Pauli encoding gates, a lot of redundancy occurs in the final frequencies. Indeed, if many eigenvalues are equal, the tree leaves will become redundant. Very small eigenvalues can also create groups of frequencies extremely close to each other, which in some use cases, when tuning their coefficients $a_\omega$ and $b_\omega$, can be considered as redundancy. In our numerical experiments (see Fig.\ref{fig:vqc_random_pauli_1d}), we observe an interesting phenomenon: On average, the frequencies with the more redundancy tend to obtain larger coefficients. Conversely, isolated frequencies are very likely to have small coefficients in comparison, making them "ghost" frequencies in $\Omega$. For VQC with solely Pauli encoding gates, we observe that the coefficients of high frequencies are almost stuck to zero during training. Therefore, one can argue that the \emph{Distinct Sampling} described above can reach even more frequencies than the corresponding VQC. However, if one wants to closely approximate a given VQC with RFF, one would not want to sample such isolated frequencies from $\Omega$ but instead draw with more probability the redundant frequencies.

\begin{algorithm}[H]
\caption{RFF with Tree Sampling}\label{alg:tree_sampling}
\begin{algorithmic}[1]
\Require a VQC model $f$, and $M$ datapoints $\{x_j\}_{j\in[M]}$
\Ensure Approximate function $\tilde{f}$
\State Diagonalize the Hamiltonians of the VQC's encoding gates.
\State Sample $D$ paths from the tree shown in Fig.\ref{fig:tree_sampling}, obtain $D$ frequencies $(\omega_1,\cdots,\omega_D)$ from $\Omega$
\State Follow steps 4-6 of Algorithm \ref{alg:distinct_sampling}.
\end{algorithmic}
\end{algorithm}

This is what we try to achieve with \emph{Tree Sampling}. Knowing the eigenvalue decomposition of each encoding's Hamiltonian, we propose to directly sample from the tree shown in Fig.\ref{fig:tree_sampling}. The first advantage of this method is that it does not require computing the whole set $\Omega$, but only draw $D$ paths through the tree (which can be used to generate up to $\binom{D}{2}+1$ positive frequencies, with potential redundancy). Second, it naturally tends to sample more frequencies that are redundant, and therefore more key to approximate the VQC's function. Overall, it could speed up the running time and necessitate fewer samples.

\subsubsection{RFF with Grid sampling}
The two above methods suffer from a common caveat: if one or more of the encoding Hamiltonians are hard to diagonalize, sampling the VQC's frequencies is not possible as it prevents us from building some of the branches of the tree shown in Fig.\ref{fig:tree_sampling}. 

Even in this case, we propose a method to approximate the VQC. If the frequencies are unknown, but one can guess an upper bound or their maximum value, we propose the following strategy: We create a grid of frequencies regularly disposed between zero and the upper bound $\omega_{max}$, on each dimension. In practice, if unknown, the value of $\omega_{max}$ can simply be the largest frequency learnable by the Shannon criterion (see Section \ref{sec:RFF_error_num_samples}) hence half the number of training points. 
Letting $s>0$ be the step on this grid, the number of frequencies on a single dimension is given by $\omega_{max}/s$. Over all dimensions, there are $\lceil(\omega_{max}/s)\rceil^d$ frequency vectors.  

Therefore, instead of sampling from actual frequencies in $\Omega$, one could sample blindly from this grid, hence the name \emph{Grid Sampling}. At first sight, it might seem ineffective, since none of the frequencies actually in $\Omega$ may be represented in the grid. But we show in the Appendix \ref{sec:gridsampling_proof} that the error between the VQC's model $f$ and the approximation $\tilde{f}$ coming from the grid can be bounded by $s$. When $s$ is small enough, the number $D$ of samples necessary to reach an error $\epsilon>0$ grows like $1/\epsilon^2log(1/s)$ which is surprisingly efficient. However, the trade-off comes from the fact that a small $s$ means a very large grid, in particular in high dimension.

\begin{algorithm}[H]
\caption{RFF with Grid Sampling}\label{alg:grid_sampling}
\begin{algorithmic}[1]
\Require Assumption on the highest frequency $\omega_{max}$, a step $s>0$ and $M$ datapoints $\{x_j\}_{j\in[M]}$
\Ensure Approximate function $\tilde{f}$
\State Create a regular grid in $[0,\omega_{max}]^d$ with step $s$.
\State Sample $D$ frequencies $(\omega_1,\cdots,\omega_D)$ from the grid.
\State Follow steps 4-6 of Algorithm \ref{alg:distinct_sampling}.
\end{algorithmic}
\end{algorithm}

\subsection{Number of Samples and Approximation Error}
\label{sec:RFF_error_num_samples}
%\jonas{Maybe it should go to the next section, or Appendix}

In \cite{sutherland2015error}, authors bound the resulting approximation error in the RFF method. We provide the main theorems in the Appendix \ref{app:rff_bounds}. In Theorem \ref{thm:approx_krr}, we consider the error between the final functions $f$ and $\tilde{f}$. We see that if the error must be constrained such that $|f(x)-\tilde{f}(x)|\leq \epsilon$, for $\epsilon>0$, one can derive a lower bound on the number $D$ of samples necessary for the approximation. Fortunately, the bound on $D$ grows linearly with the input dimension $d$, and logarithmically with $\sigma_p$ which linked to the variance of the frequency distribution $p(\omega)$.

In our case, the continuous distribution $p(\omega)$ will be replaced by the actual set of frequencies $\Omega$, 
\begin{equation}
    p(\omega) = \frac{1}{|\Omega|}\Sum_{\omega \in \Omega} \delta_\omega
\end{equation}
where $\delta_\omega$ represents the Dirac distribution at $\omega$. As a result, we can write the discretized variance as
\begin{equation}
    \sigma_p=\Sum_{\omega \in \Omega} p(\omega) \omega^T\omega
\end{equation} 
From this, we want to know the link between the number $D$ of samples necessary and the size of $\Omega$, or even the number $L$ of encoding gates per dimension.

In the general case, we consider that $\sigma_p$ is the average value of $\omega^T\omega$, that is to say the trace of the multidimensional variance. The more the frequencies are spread, the higher $\sigma_p$ will be, but the number of frequencies in itself doesn't seem to play the key role here.

Finally, note that it is important to take into account the Shannon criterion, stating that one needs at least $2\omega_{\max}$ training points to estimate the coefficients of a Fourier series of maximum frequency $\omega_{\max}$.
In practice, it puts some limitation on the largest frequency one can expect to learn (both classically and quantumly) given an input dataset. Exponentially large frequencies with VQCs, as in \cite{shin2022exponential} would have a limited interest with subexponential number of training points.
The efficiency of RFF against VQCs in such cases becomes even more interesting, as it allows reducing the number $D$ of sample compared to the actual exponential size of $\Omega$.

%It follows that $O(log(\sigma_p))=O(log(|\Omega|))$. 

\subsubsection{Pauli encoding}
We provide here a bound on the minimum of samples required to achieve a certain error between the RFF model and the complete model in the case of Pauli encoding in the distinct sampling strategy. The proof and details for this theorem is shown in Appendix \ref{app:proof_pauli}.

\begin{theorem}
\label{thm:number_of_samples_pauli_encoding}
Let $\mathcal{X}$ be a compact set of $\mathbb{R}^d$, and $\epsilon > 0$. We consider a training set $\{(x_i, y_i)\}_{i=1}^M$. Let $f$ be a VQC model with $L$ encoding Pauli gates on each of the $d$ dimensions and full freedom on the associated frequency coefficients, trained with a regularization $\lambda$. Let $\sigma_y^2 = \frac{1}{M}\sum_{i=1}^M y_i^2$ and $|\mathcal{X}|$ the diameter of $\mathcal{X}$. Let $\tilde{f}$ be the RFF model with $D$ samples in the distinct sampling strategy trained on the same dataset and the same regularization. Then we can guarantee $|f(x) - \tilde{f}(x)| \leq \epsilon$ with probability $1 - \delta$ for a number $D$ of samples given by: 
\begin{equation}
\label{eq:theorem_samples_vs_error_Pauli}
%\begin{split}
    D = \Omega\Bigg(\frac{dC_1(1+\lambda)^2}{\lambda^4\epsilon^2}\bigg[log(dL^2 |\mathcal{X}|)%\\
    + log\frac{C_2(1+\lambda)}{\epsilon\lambda^2} - log\delta\bigg]\Bigg)
%\end{split}
\end{equation}
with $C_1$, $C_2$  being constants depending on $\sigma_y$, $|\mathcal{X}|$. We recall that in Eq.\ref{eq:theorem_samples_vs_error_Pauli} the notation $\Omega$ stands for the computational complexity "Big-$\Omega$" notation. 
\end{theorem}

We can conclude that the number $D$ of samples grows linearly with the dimension $d$, and logarithmically with the size of $\Omega$. It means that even though the number of frequencies in the spectrum of the quantum circuit is high, only a few of them are useful in the model. This fact limits the quantum advantage of such circuits. However, the scaling in $\epsilon$ and $\lambda$, respectively in $\Omega(1/\epsilon^2)$ and $\Omega(1/\lambda^4)$ is not favorable, and can limit in practice the use of the RFF method.

 One may think at first sight that the RFF method would be efficient to approximate the outputs of any VQC, or find a "classical surrogate"\,\cite{schreiber2022classical}. Instead, the bound that is provided is on the error between a VQC trained on an independent data source and a RFF model trained on the same data. There is contained in this result the potential incompleteness of the dataset to render an accurate representation of the underlying data distribution. If the dataset fails to correctly represent the data distribution, then the VQC will fail to correctly model it, and the theorem provide the minimal number of samples to perform "as badly", and this number could be low.

This was tested in the numerical simulations in Section \ref{sec:exp_numerics}. However, we recall that these are bounds that we cannot reach in practice with the current classical and quantum computing resources. They give an intuition on the asymptotic scaling as quantum circuits become larger.

Note that this proof changes if we take into account the redundancy of each frequency  when sampling from $\Omega$. This will be the case in the \emph{Tree Sampling} strategy (see Section \ref{sec:tree_sampling}). In that case, the variance becomes even smaller since some frequencies are more weighted than others, in particular for Pauli encoding.

\subsubsection{Grid sampling}
We provide here a bound on the minimum number of samples required to achieve a certain error between the RFF model and the complete model in the case of a general encoding in the gird sampling strategy. The proof and details for this theorem is shown in Appendix \ref{sec:gridsampling_proof}. 

\begin{theorem}
\label{thm:number_of_samples_grid_sampling}
Let $\mathcal{X}$ be a compact set of $\mathbb{R}^d$, and $\epsilon > 0$. We consider a training set $\{(x_i, y_i)\}_{i=1}^M$. Let $f$ be a VQC model with any hamiltonian encoding, with a maximum individual frequency $\omega_{max}$ and full freedom on the associated frequency coefficients, trained with a regularization $\lambda$. Let $\sigma_y^2 = \frac{1}{M}\sum_{i=1}^M y_i^2$ and $|\mathcal{X}|$ the diameter of $\mathcal{X}$. Let $\tilde{f}$ be the RFF model with $D$ samples in the grid strategy trained on the same dataset and the same regularization. Let $C = |f|_\infty |\mathcal{X}|$ and $s$ the sampling rate defined in the grid sampling strategy. Then we can guarantee $|f(x) - \tilde{f}(x)| \leq \epsilon$ for $0 <s<\frac{1}{C}$  with probability $1 - \delta$ for a number $D$ of samples given by: 
\begin{equation}\label{eq:theorem_samples_vs_error_Grid}
\begin{split}
    D = \Omega\Bigg(\frac{dC_1(1+\lambda)}{\lambda^4(\epsilon - sC)^2}\bigg[log(\omega_{max} |\mathcal{X}|)\\
    + log\frac{C_2(1+\lambda)}{\lambda^2(\epsilon - sC)} - log\delta\bigg]\Bigg)
\end{split}
\end{equation}
with $C_1$ and $C_2$ being constants depending on $\sigma_y$ and $d(X)$. We recall that in Eq.\ref{eq:theorem_samples_vs_error_Grid} the notation $\Omega$ stands for the computational complexity "Big-$\Omega$" notation. 
\end{theorem}

%\section{Study of $\Omega$ and sample complexity}

%\subsection{Warm-up: Pauli encodings}

%\subsection{Dimension of the datapoints $\X$}

\subsection{Limitations of RFF for Approximating VQCs}
\label{sec:limitations_QuantumRFF}

We now have seen the theoretical power of Random Fourier Features and three different adaptations to approximate VQCs in practice. Since many parameters are to be taken into account (size and structure of $\Omega$, number of qubits, circuit depth, number of training points, input dimension,  encoding Hamiltonians, etc.), it is natural to ask ourselves which of the three strategy is recommended given a use case, and are there any use cases for which none of them work.

As seen in Section \ref{sec:RFF_error_num_samples}, we know the lower bound on the number of samples to draw in RFF, to reach a specific error. This bound grows linearly with the input dimension $d$, and logarithmically with the size of $\Omega$ (itself depending exponentially on $L^d$). Nonetheless, in practice, we could see that very large spectrum are hard to approximate, simply because it would require much more samples. This scaling will be judged once such VQCs will be actually implemented on large enough quantum computers (with enough qubits and/or long coherence). 

The size of $\Omega$ increases as well when the encoding Hamiltonians have distinct eigenvalues and are acting on many qubits. Therefore, quantum computers allowing for many qubits and various high locality Hamiltonians would be a plus for enlarging the spectrum.

As the Hamiltonians become larger and their eigenvalues complex, we could reach a limit where it becomes impossible to diagonalize them. In such a case, without sampling access to $\Omega$, the \emph{Distinct} and \emph{Tree} sampling strategies would be unavailable. The \emph{Grid} sampling scheme would suffice until suffering from the high dimensionality or other factors detailed above.

Another important limitation of the theoretical bounds is the scaling in the regularization parameter. As illustrated in figure \ref{fig:impact_D_and_L}, this constant prefactor make the bound still reach intractable numbers, so it is of little value to determine \textit{a priori} if the RFF method will be efficient.

Finally, having a small dataset would limit the trainability of our classical RFF methods. Note that this would probably constrain the training of the VQC as well.

Overall, some limits for our classical methods can be guessed and observed already, but the main ones remain to be measured on real and larger scale quantum computers. We leave this research for future work. On another hand, one could want to understand better the relation between the available frequencies and their amplitude in practice, to find potential singularities that could help, or not, the VQCs.\\

\begin{figure*}[t!]
    \centering
    \includegraphics[width=0.6\textwidth]{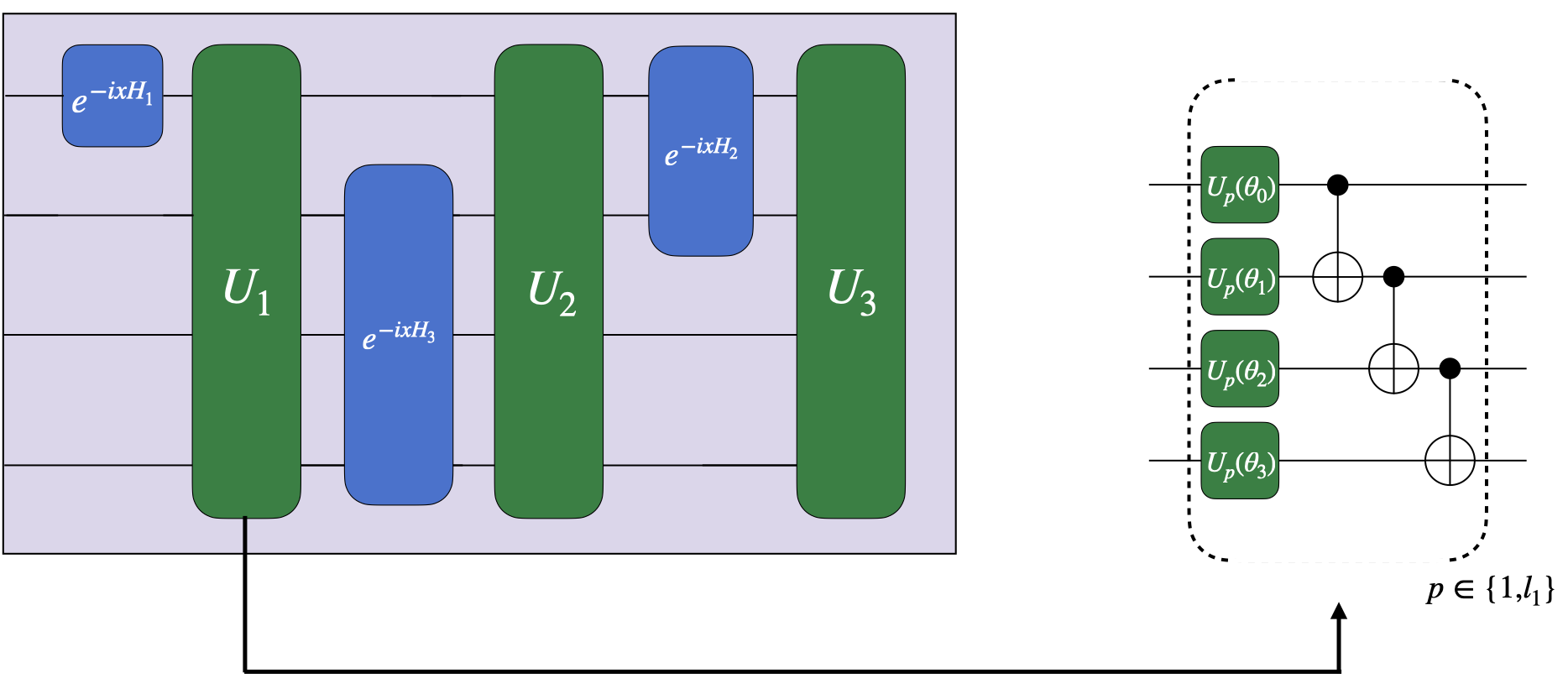}
    \caption{\textbf{Random instance of a VQC.} In this example, three encoding Hamiltonians $\{H_1, H_2, H_3 \}$ are randomly assigned over four qubits, and load a 1-dimensional vector $x$. Following each encoding gate $H_i$, an ansatz with trainable parameters and a ladder of CNOTs is applied, $l_i$ times in a row.}
    \label{fig:vqc_random_architecture}
\end{figure*}

Finally, we want to insist on the fact that the assumptions on VQCs are crucial on the whole construction that we propose, and that some of them could be questioned, especially concerning the encoding. For instance, when encoding vectors $x=(x_1,\cdots,x_d)$, not having encoding gates expressed as $exp(-x_iH)$ could potentially change the expression of $f(x;\theta)$ (Eq.\ref{eq:f_as_fourier_serie_complex}) and therefore could change the fact that the associated kernel would be easily expressed as a Fourier series, with shift-invariance. For instance, in \cite{kyriienko2021solving}, the authors use $exp(-arcsin(x_i)H)$ to encode data, resulting in $f$ being expressed in the Chebyshev basis instead of the Fourier one. More generally, understanding what happens with encodings of the form $exp(-g(x_i)H)$, and whether we can still use our classical approximation methods, remain an open question. Similar questions arise if we use simultaneous components encoding $exp(-x_ix_jH)$, or other alternative schemes.

% \begin{figure}[h!]
%   \begin{subfigure}{\linewidth}
%          \centering
%          \includegraphics[width=0.8\linewidth]{figures/random_VQC_a.pdf}
%          \caption{Typical instance of a random VQC }
%          \label{fig:vqc_random_architecture_a}
%      \end{subfigure}
%      \hfill
%      \begin{subfigure}{\linewidth}
%          \centering
%          \includegraphics[width=0.6\linewidth]{figures/random_VQC_b.pdf}
%          \caption{Decomposition of an ansatz layer.}
%          \label{fig:vqc_random_architecture_b}
%      \end{subfigure}
%     \caption{\textbf{Random instance of a VQC.} In this example, three encoding Hamiltonians $\{H_1, H_2, H_3 \}$ are randomly assigned over four qubits, and load a 1-dimensional vector $x$. Following each encoding gate $H_i$, an ansatz with trainable parameters and a ladder of CNOTs is applied, $l_i$ times in a row.}
%         \label{fig:vqc_random_architecture}
% \end{figure}

\section{Experiments and Numerics}
\label{sec:exp_numerics}

%\subsection{PyQTorch: a module for GPU emulation of quantum models}
%\jonas{(for @Slimane)} 

%A fast large scale emulator for quantum machine learning on a pytorch backend 

In this section, we aim to assess the accuracy and efficiency of our classical methods to approximate VQCs in practice. We have conducted three types of simulations. In Section \ref{sec:exp_rff_on_vqc_random}, we instantiate random VQCs and try to mimic their output using our RFF methods. In Section \ref{sec:exp_rff_on_artificial_functions}, we create artificial functions and compare VQCs and RFF on the same task of learning it. In Section \ref{sec:exp_rff_real_datasets}, we compare VQCs and RFF on real datasets. Finally, we observe the scaling of our methods in Section \ref{sec:exp_RFF_scaling_with_Omega}.

As shown in Fig.\ref{fig:vqc_random_architecture}, a typical random VQC instance is built from a list of general encoding Hamiltonians $\{H_1, \cdots, H_k \}$, applied to randomly selected qubits according to their locality. The number of qubits is fixed to 5 in the following experiments (Note that the number of qubits has no impact on the expressivity a priori, it will only influence the spans of coefficients).

\subsection{Using RFF to Mimic Random VQCs}
\label{sec:exp_rff_on_vqc_random}

In a first stage, we focus on approximating random VQCs (i.e. with random parameter initialization) using Random Fourier features.
To this end, we fix the quantum spectrum $\Omega$ by making a certain choice about the structure of the encoding gates. We prepare the training dataset $\{X_{grid}, Y_{grid} \}$ with $X_{grid}$ being a set of $d$-dimensional data points spaced uniformly on the interval %$[0,x_{max}]^d$
$\prod_{i=1}^d [0,x_{max_i}]$ and $Y_{grid}$ the evaluation of the quantum circuit on the input dataset $X_{grid}$. We, then, observe and evaluate the performance of our three RFF strategies (\emph{Distinct}, \emph{Tree}, and \emph{Grid} sampling, see Section \ref{sec:3sampling_strategies_RFF}) to approximate the quantum model. For completeness, we have tested our methods on different types of VQCs: some with basic Pauli encoding in 1 dimension (Fig.\ref{fig:vqc_random_pauli_1d}), in higher dimension (Fig.\ref{fig:vqc_random_pauli_d}), some with more complex Hamiltonians (Fig.\ref{fig:vqc_random_hamiltonian_complex}), and with scaled Pauli encoding as in \cite{shin2022exponential} (Fig.\ref{fig:vqc_random_scaled_pauli_exp}). For each type, we have also observed the actual Fourier Transform of the random VQCs model on average, to understand which frequencies appear more frequently in their spectrum.

\begin{figure}[h!]
\centering
\includegraphics[width=0.95\linewidth]{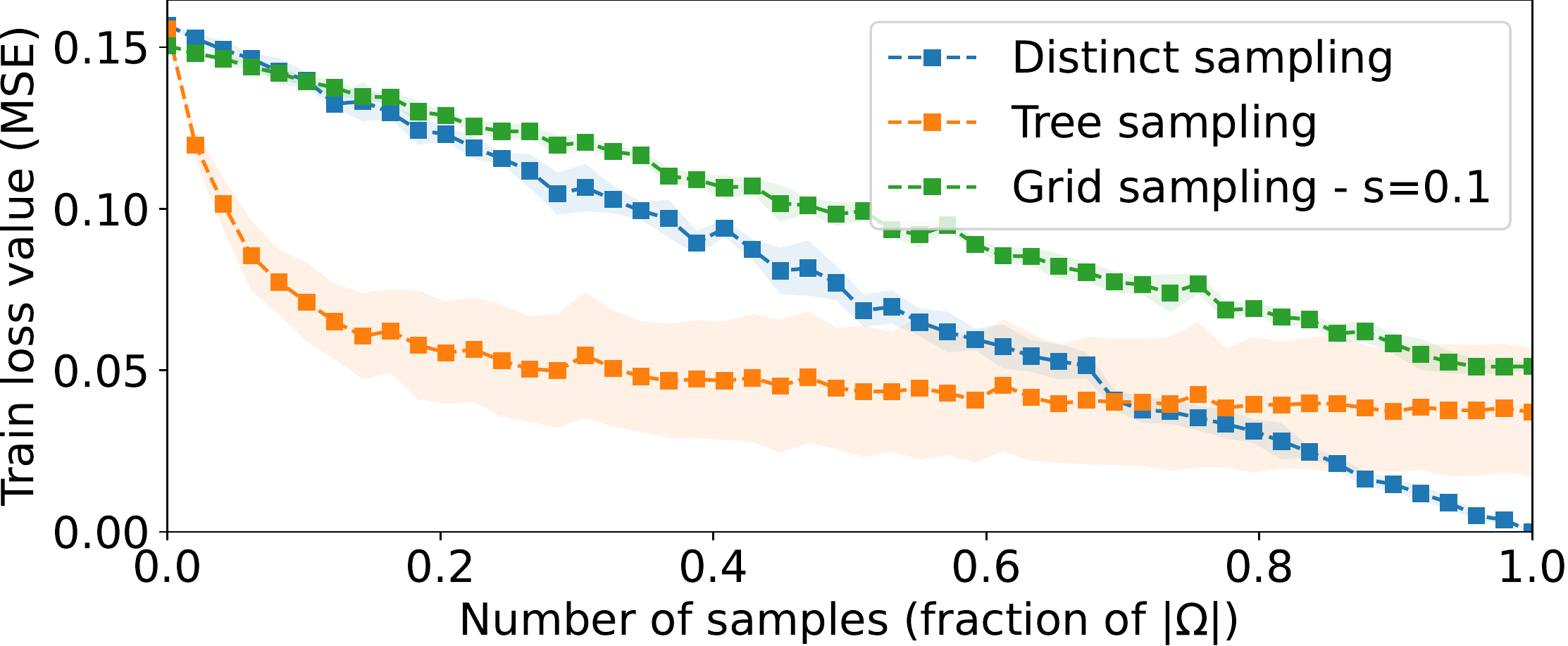}
\caption{RFF performance for $L=5, d=4$, to approximate random VQCs with Pauli encoding.}
\label{fig:vqc_random_pauli_d}
\end{figure}

% \begin{figure}[h!]
%      \begin{subfigure}{\linewidth}
%          \centering
%          \includegraphics[width=\linewidth]{figures/Fig1_a_100.png}
%          \caption{Fourier Transform of the quantum models obtained with random VQCs. On average, the frequencies with high coefficients are the one with high redundancy in $\Omega$ (seen in the inner red histogram). Frequencies over $100$ have negligible coefficients and redundancy, and therefore are not shown.}
%          %\helarq{($x_{max} = 50$ for the fourier transform + 20000 samples are used to plot the redundancy hist. d=1)}}
%          \label{fig:ft_200}
%      \end{subfigure}
%      \vfill
%      \begin{subfigure}{\linewidth}
%          \centering
%          \includegraphics[width=\linewidth]{figures/Fig1_b.png}
%          \caption{Evolution of RFF train loss as a function of the relative number of frequencies sampled.}
%          \label{fig:rff_200}
%      \end{subfigure}
%         \caption{Random VQCs with L=200 Pauli encoding gates, averaged over 10 different random initialization. Inputs are 1-dimensional.}
%         \label{fig:vqc_random_pauli_1d}
% \end{figure}

\begin{figure*}[t!]
  \centering
  \subcaptionbox{Average Fourier Transform of the VQC's quantum models. The frequencies with high coefficients are the ones with high redundancy in $\Omega$ (seen in the inner red histogram). Frequencies over $100$ have negligible coefficients and redundancy, and therefore are not shown.}[.45\textwidth][c]{%
    \includegraphics[width=0.99\linewidth]{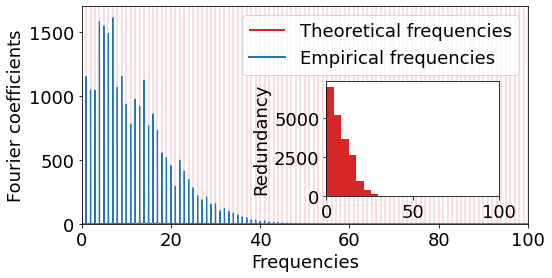}}
    \label{fig:ft_200}
    \quad
  \subcaptionbox{Evolution of RFF train loss as a function of the relative number of frequencies sampled. The \emph{Tree} sampling strategy takes advantage of the high redundancy to sample less frequencies to reach a good approximation.}[.45\textwidth][c]{%
    \includegraphics[width=0.99\linewidth]{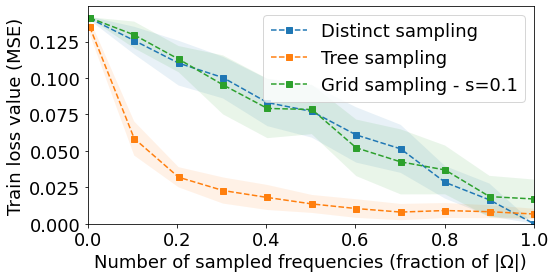}}
    \label{fig:rff_200}
\caption{Random 1d VQCs with L=200 Pauli encoding gates, averaged over 10 different random initialization.}% Inputs are 1-dimensional.}
\label{fig:vqc_random_pauli_1d}    
\end{figure*}

We note that the number of data points in $X_{grid}$ needed to efficiently learn the quantum function is $N>\prod_{i=1}^d \frac{x_{max_i} w_{max_i}}{\pi} $. This choice is basically related to the \textit{Shannon criterion} for effective sampling in order to reconstruct the full function covering all of its frequencies. Moreover, for the solution to be unique and hence for the least square problem introduced in Eq.\ref{eq:least_square} to be well defined, we better choose $N$ to be bigger than the number of features in the regression problem (these two criteria coincide in the case of Pauli encoding).

\subsubsection{Pauli encoding}
\label{numerics_a}
We first consider a quantum model with $L$ Pauli encoding gates per feature resulting in an integer-frequency spectrum (half of $\llbracket -L, L \rrbracket ^d$). In this case, the corresponding quantum model is a periodic function of period $T=(2\pi)^d$ and thus, we choose $x_{max}  = 2\pi$ for $X_{grid}$ construction.

In Fig.\ref{fig:vqc_random_pauli_1d}, we implement a VQC with $L$=200 Pauli encoding gates, for a 1-dimensional input. We observe that our classical approximation methods are indeed able to reproduce such VQCs. On average, the RFF training  error for \emph{Distinct} and \emph{Grid} sampling is a linear function of the number $D$ of samples taken from $\Omega$. On the other hand, the error using \emph{Tree} sampling exhibits a faster decreasing trend, reaching relatively low errors with only $20\%$ of the spectrum size. %Indeed, the redundancy of Pauli encoding is extremely high, since with $L=200$ gates, $\Omega$ can potentially have $3^{200}$ frequencies, but only have 200 distinct ones, concentrated in the lower part.  

We conjecture that the efficiency of \emph{Tree} sampling is closely related to the redundancy in the discrete frequency distribution over $\Omega$. In fact, as shown in Fig.\ref{fig:vqc_random_pauli_1d}, Fourier coefficients of the VQC are, on average, correlated to the frequency redundancy in the empirical quantum spectrum. 
Frequencies above a certain threshold $\omega_{effective}$ are merely redundant for this particular encoding scheme, and we observe that they are cut from the quantum model empirical spectrum. The effective spectrum of the VQC is therefore smaller than what the theory predicts. 
Consequently, the fast decreasing trend of the \emph{Tree} sampling stems from the fact that we sample according to the redundancies, therefore requiring less frequency samples. We see that $0.2 \times |\Omega|$ samples are sufficient to sample approximately all frequencies in $[|0,\omega_{effective}|]$.
%The same behaviors are observed for higher dimensional datasets with tree sampling keeping the best approximation performance \ref{fig:10d}.
In Fig.\ref{fig:vqc_random_pauli_d}, we show a similar simulations with a $d$-dimensional input ($d=4$) and $L=5$ Pauli gates per dimension. According to Eq.\ref{eq:size_of_Omega_plus}, the theoretical number of distinct positive frequencies is $7321$. In this case in the tree sampling procedure, we can sample both a frequency and its opposite without removing one of them. Therefore the scheme is a bit less performant than in dimension 1.

\subsubsection{More complex Hamiltonian encodings}

\begin{figure*}[!t]
  \centering
  \subcaptionbox{Fourier transform averaged over different random initialization of the VQC. The intensity of the vertical red lines indicates the concentration of the theoretical frequencies in $\Omega$.}[.45\textwidth][c]{
    \includegraphics[width=0.99\linewidth]{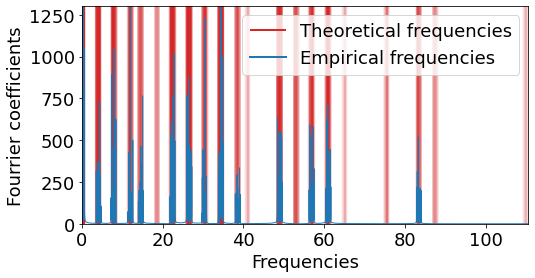}}
    \label{fig:exp_vqc_random_hamiltonian_fourier_transform}
    \quad
  \subcaptionbox{RFF train loss with different sampling methods on the random VQCs. The Distinct sampling benefit from the concentration of frequencies in packets to approximate with less samples.}[.45\textwidth][c]{
    \includegraphics[width=0.99\linewidth]{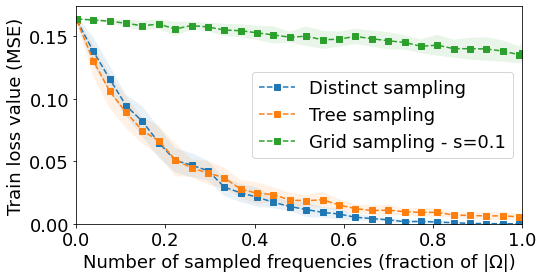}}
    \label{fig:exp_vqc_random_hamiltonian_RFF_loss}
\caption{Random 1d VQC with 4 scaled Paulis and and a 3-qubits $H_{XYZ}$ Hamiltonian}
\label{fig:vqc_random_hamiltonian_complex}    
\end{figure*}

For Pauli encoding, we have seen that \emph{Tree} sampling is highly effective for approximating the quantum model. Consequently, we designed VQCs with different spectrum distributions to study the RFF approximation performance in these cases.

%\hela{When dealing with non integer-frequency quantum models, the quantum function is no longer $2\pi$ periodic and hence the choice of $x_{max}$ for constructing $\chi_{random}$ should be done delicately. To determine the appropriate $x_{max}$, we follow an experimental procedure where we plot the performance of RFF1 for different values of $x_{max}$ and determine the threshold above which the approximation performance becomes independent of $x_{max}$.}

As explained in Section \ref{sec:VQC_preliminaries}, we consider encoding gates of the form $exp(-ix_iH)$ for each dimension $i$. One way to alter the spectrum distribution is the use of more general Hamiltonians $H$. To obtain exotic Hamiltonians while maintaining their physical feasibility (involving only two-bodies interactions), we use the generic expression
\begin{equation}\label{eq:Hxyz}
    H_{XYZ} = \sum_{\langle i,j\rangle} \alpha_{ij}X_iX_j
+
\beta_{ij}Y_iY_j
+
\gamma_{ij}Z_iZ_j
+
\sum_{i} \delta_i P_i
\end{equation}
with the first term describing the interactions: $\langle i,j\rangle$ indicates a pair of connected particles and the second term describing a single particle's energy ($P_i =$ \{$X_i$, $Y_i$ or $Z_i$\}). 

In Fig.\ref{fig:vqc_random_hamiltonian_complex}, we construct VQCs mixing both such Hamiltonians\footnote{in Fig.\ref{fig:vqc_random_hamiltonian_complex}, we used a 3-qubits Hamiltonian defined by: $H_{XYZ} =7 X_0 X_1 + 7 X_1 X_0 + 0.11 X_0 X_2 + 0.1 X_2 X_0 + 8 [ Y_1 Y_2 + Y_2 Y_1 + Z_0 Z_2 + Z_2 Z_0]$} and scaled Paulis\footnote{Scaling factors are $[26.4309,34.4309,22.4309,0.4309]$} as encoding gates, on 1-dimensional inputs. 
In these cases, the corresponding quantum model is no longer $2\pi$-periodic, thus we have to find empirically a good value for $x_{max}$ (by increasing it until the performance reaches a limit). 

%Once again, we observe that low redundant frequencies are cut from the empirical spectrum (Fig.\ref{fig:exp_vqc_random_hamiltonian_fourier_transform}). As a consequence, the \emph{Tree} sampling method has once again a sublinear decreasing scaling. 

With such complex encoding, we witness a different behavior for the \emph{Distinct} sampling method, in  comparison to the previous basic Pauli encoding scenario. Essentially, \emph{Distinct} sampling has a faster than linear scaling, showing a clear and unexpected efficiency of RFFs in this case. We also notice that the \emph{Tree} sampling method has a similar scaling. 
This observation points to the fact that, with the chosen encoding strategy, the frequencies in the spectrum $\Omega$ are concentrated in many packets or groups. This behavior is displayed with the concentrated red lines in Fig.\ref{fig:vqc_random_hamiltonian_complex}. Therefore, even though the frequencies in $\Omega$ have a low redundancy (545 distincts frequencies out of 2017), sampling just one of the many frequencies in a narrow packet is enough for the RFF to approximate it all. To put it diffrently, we can consider that there is a $\Omega_{effective}$ where each packet can be replaced by its main frequency, and the RFF manage to approximate it with fewer samples than the actual size of $\Omega$. To conclude, many distinct frequencies is not a guarantee of high expressivity. 

As for \emph{Grid} sampling, the choice of $s$ seemed too high for this solution to work in this case, in line with the theoretical bounds for this sampling method given in Theorem \ref{thm:number_of_samples_pauli_encoding}.

% \begin{figure}
%      \begin{subfigure}{\linewidth}
%          \centering
%          \includegraphics[width=\linewidth]{figures/fig8_ft.png}
%          \caption{Fourier transform averaged over different random initialization of VQC$_{3}$. The intensity of the vertical red lines indicates the redundancy and/or concentration of the theoretical frequencies in $\Omega$.}
%          \label{fig:exp_vqc_random_hamiltonian_fourier_transform}
%      \end{subfigure}
%      \hfill
%      \begin{subfigure}[b]{\linewidth}
%          \centering
%          \includegraphics[width=\linewidth]{figures/Fig2_b.png}
%          \caption{RFF train loss with different sampling methods on the random VQC$_{3}$.}
%          \label{fig:exp_vqc_random_hamiltonian_RFF_loss}
%      \end{subfigure}
%         \caption{1-dimensional VQC with 4 scaled Paulis and and a 3-qubits $H_{XYZ}$ Hamiltonian}
%         %\caption{1-dimensional VQC with 4 scaled Paulis and and a $H_{XYZ}$ Hamiltonian with the following pairs of connected qubits $\{\langle 0,1\rangle,\langle 0,2\rangle,\langle 1,2\rangle\}$ and with $\delta_i=0 \forall i$. For this particular circuit, We obtain $\omega_{max}=110.5$, $|\Bar{\Omega}|=545$ and $|\Omega|=2017$.}
%         \label{fig:vqc_random_hamiltonian_complex}
% \end{figure}

\subsubsection{Exponential Pauli Encoding}

In order to obtain VQCs with a large number of frequencies, but low redundancy and no concentrated packets, we exploit the exponential encoding scheme proposed in \cite{shin2022exponential}, resulting in a non degenerate quantum spectrum with zero redundancies and thus a uniform probability distribution over integers. In this encoding strategy, encoding Pauli gates are enhanced with a scaling coefficient $\beta_{kl}$ for the $l^{th}$ Pauli rotation gate encoding the component $x_k$. This gives us a total of $3^{Ld}$ positive and negative frequencies. These frequencies can be all distinct with the particular choice of $\beta_{kl} = 3^{l-1}$, resulting in an exponentially large and uniform $\Omega$. Note however that $\Omega$ is analytically known and contain only integer frequencies, mostly very high frequencies for which the usefulness in practice remain to be studied. 

We have tested our classical RFF approximation, shown in Fig.\ref{fig:vqc_random_scaled_pauli_exp}, and obtained again the confirmation that RFF can approximate such an exponential feature space with a fraction of $|\Omega|$. This fraction might however be too large in practice. We also observe as expected that all three strategies have a linear scaling, in line with the absence of redundancy and frequency packets. 

%\jonas{[PENDING SLIMANE EXP] Since $|\Omega| = O(3^{Ld})$, and since according to Theorem \ref{thm:number_of_samples_pauli_encoding} the lower bound on $D$ should grow as $O(d, log(|\Omega|)$, we can increase exponentially $|\Omega|$ without changing $d$, by increasing $L$ only. In Fig.\ref{TO DO}, we plot the fraction of samples necessary to RFF reach a certain error, in function of an exponentially growing $\Omega$, at fixed dimension. We observe that... [To Be Completed]}

% \begin{figure}
%   \begin{subfigure}{\linewidth}
%          \centering
%          \includegraphics[width=\linewidth]{figures/Fig6_a.png}
%          \caption{Fourier transform averaged over different random initializations of the exponential encoding VQC with $L=5$. }
%          \label{fig:vqc_random_scaled_pauli_exp_FT}
%      \end{subfigure}
%      \hfill
% %     \begin{subfigure}{\linewidth}
% %         \centering
% %         \includegraphics[width=\linewidth]{figures/exponential5_ft_18.png}
% %         \caption{Fourier transform of an instance of the exponential encoding VQC with $L=5$.}
% %         \label{fig:todo4}
% %     \end{subfigure}
% %     \hfill
%      \begin{subfigure}[b]{\linewidth}
%          \centering
%          \includegraphics[width=\linewidth]{figures/exponential5.png}
%          \caption{RFF approximation performance of an the exponential encoding VQC with $L=5$.}
%          \label{fig:vqc_random_scaled_pauli_exp_RFF}
%      \end{subfigure}
   
%         \caption{VQC1:  \jonas{CHANGER. remplacer cet exp par les scaled pauli du paper $3^L$}}
%         \label{fig:vqc_random_scaled_pauli_exp}
% \end{figure}

\begin{figure*}[t!]
  \centering
  \subcaptionbox{Fourier transform averaged over different random initializations of the exponential encoding VQC with $L=5$. }
  [.45\textwidth][c]{%
    \includegraphics[width=0.99\linewidth]{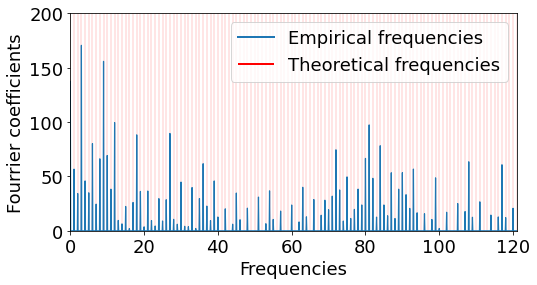}}
    \label{fig:vqc_random_scaled_pauli_exp_FT}
    \quad
  \subcaptionbox{RFF approximation performance of an the exponential encoding VQC with $L=5$.}
  [.45\textwidth][c]{%
    \includegraphics[width=0.99\linewidth]{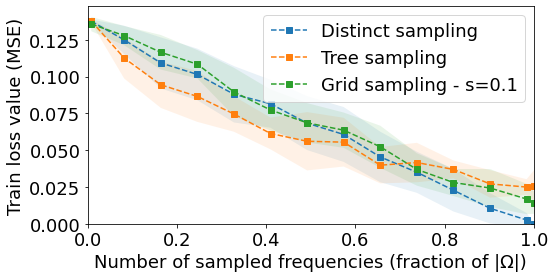}}
    \label{fig:vqc_random_scaled_pauli_exp_RFF}
\caption{Random VQCs with exponentially large spectrum, using scaled Pauli encoding as in \cite{shin2022exponential}.}
\label{fig:vqc_random_scaled_pauli_exp}    
\end{figure*}

With these experiments, we conclude a few important properties. We observe that when some frequencies in the spectrum $\Omega$ have much redundancy (\textit{e.g.} Pauli encoding), these frequencies are empirically the ones with higher coefficients. In such case, the \emph{Tree} sampling strategy is able to approximate the VQC's model with fewer samples than the other methods as expected. With more complex Hamiltonians, concentrated packets of frequencies appear, and even without much redundancy, both \emph{Tree} and \emph{Distinct} sampling require fewer frequency samples to cover these packets. According to these experiments, the worst case scenario for the RFF is a uniform probability distribution where all the three sampling techniques will be equivalent. 
Nonetheless, the theoretical bounds prove that the number of Fourier Features will scale linearly with respect to the spectrum size that scales itself exponentially.

%\hela{\paragraph*{Conclusion:}These different experiments showcase the importance of the spectrum size combined with its corresponding probability distribution on the RFF approximation performance of the quantum model. In fact, by studying the characteristics of the discrete probability distribution over the spectrum $\Omega$ (redundancies compared to the total number of discrete frequencies, concentration of frequencies in narrow packets, etc \dots), one can gain potential insights into the effectiveness of the RFF in approximating the quantum model. According to these experiments, the worst case scenario for the RFF is a uniform probability distribution where all the three sampling techniques will be equivalent. Nonetheless, the theoretical bounds prove that the number of Fourier Features will scale linearly when the spectrum size scales exponentially.}

\subsection{Comparing VQC and RFF on Artificial Target Functions}
\label{sec:exp_rff_on_artificial_functions}

In the above section, we trained a RFF to mimic the output of random VQCs. In practice, the ground truth relies on a classical dataset with an underlying function to find. A more practical comparison is to measure the efficiency of both VQC and RFF on a common target function. We want to see when RFF can obtain a similar or better results than the VQC on the same task. 

We have seen that for VQCs with Pauli encoding, the Fourier coefficients are rapidly decreasing, cutting out frequencies higher than $\omega_{effective}$ from the empirical spectrum. 
For this reason, we have chosen a particular synthetic target function: we create a sparse Fourier series  (i.e. having only few non-zero coefficients) as a target function: $s(x) = \sum_{\omega \in \{4,10,60\}} cos(\omega x) + sin(\omega x)$ and a VQC with $L=200$ Pauli encoding gates as the quantum model. 

In Fig.\ref{fig:artificial_targe_functiont_results}, we clearly see that the VQC, as well as RFF with \textit{Tree} sampling, can not learn the frequency $\omega = 60 > \omega_{effective}$ (their train loss reach a high limit) while the RFF models based on \textit{Distinct} and \textit{Grid} sampling can effectively learn the target function with enough frequency samples. This result shows that even when a VQC with Pauli encoding %\hela{(or its equivalent \textit{Tree} sampling RFF model)}
is trained, it cannot reach all of its theoretical spectrum, thus questioning the expressivity of such a quantum model. On the other hand, %whereas 
its classical RFF approximators %\hela{not taking into account frequencies redundancy} 
(\emph{Distinct} and \emph{Grid})
succeed in learning the function in this case.
This is due to the specific choice of the frequencies in $s$ and of the VQC's structure, which has a high redundancy in its spectrum. Indeed, the VQC is not able in practice to obtain non-negligible coefficients for frequencies higher than $\omega_{effective}$. This limit appears similarly for the \emph{Tree Sampling} RFF, but not for the two other types.
%\hela{We can say that a low frequency redundancy limits the degrees of freedom with which its corresponding coefficient can be tuned in the quantum model which is clearly intuitive in its associated \textit{Tree} sampling RFF model.}
Of course, when we chose an artificial function for which all frequencies are below $\omega_{effective}$, the VQC and all RFF methods manage to fit it.
%\helarq{More examples are included in Appendix ..  where all target functions frequencies are below or above $\omega_{effective}$  }.
%\helarq{Do we want to focus more on the fact that the RFF is approximating the trained quantum model without even having access to VQC output values because optimizing the VQC is equivalent to optimizing rff with tree sampling as depicted in Fig.(\ref{fig:Fig_1}) or that it can outperform it?}

\begin{figure*}
  \centering
  \subcaptionbox{Predictions of the target function $s(x)$ with the quantum model and its corresponding RFF approximator using \textit{Distinct} sampling on 80\% of all frequencies.}
  [.5\textwidth][c]{%
    \includegraphics[width=\linewidth]{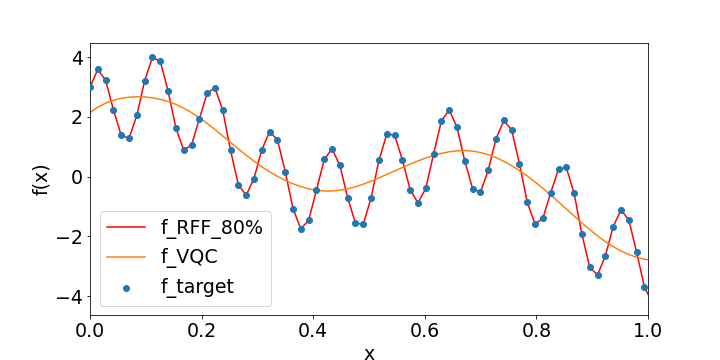}}
    \label{fig:functions(x)}
    \quad
  \subcaptionbox{RFF learning curves for sparse target fitting based on the VQC description. Distinct and Grid sampling are able to outperform the VQC.}
  [.45\textwidth][c]{%
    \includegraphics[width=0.95\linewidth]{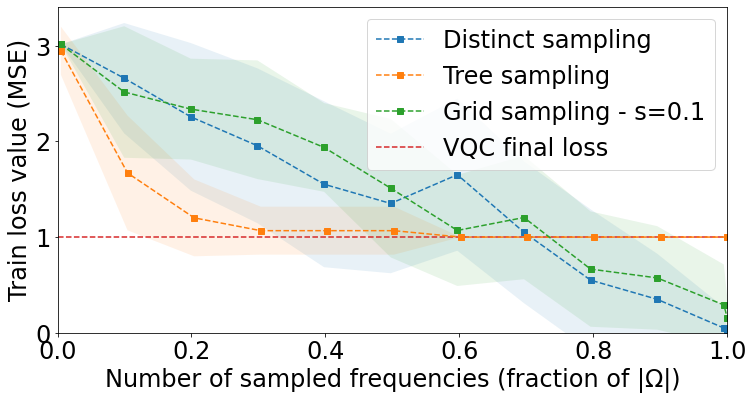}}
    \label{fig:rff_s1}
\caption{\textbf{Fitting a target function} $s(x) = \sum_{\omega \in \{4,10,60\}} cos(\omega x) + sin(\omega x)$ with a VQC architecture of $L=200$ Pauli gates.}
\label{fig:artificial_targe_functiont_results}   
\end{figure*}

\subsection{Comparing VQC and RFF on Real Datasets}
\label{sec:exp_rff_real_datasets}

In order to compare the learning performances of a VQC and its corresponding RFF approximator on real-world data, we choose the fashion-MNIST dataset \cite{xiao2017fashion} and consider a binary image classification task(\emph{coat} and \emph{dress}). We pre-process the input images by performing a principal component analysis keeping the first 5 features (therefore $d=5$) and by re-scaling the new input features between $-\pi$ and $\pi$. We also use the California Housing dataset for a regression task with 5 features and a re-scaling between 0 and $\pi$.

We chose to solve these two problems by training VQCs with Pauli encoding ($L=5$ for each dimension). According to Eq.\ref{eq:size_of_Omega_plus}, the number of distinct positive frequencies in $\Omega$ is 80526. In Fig.\ref{fig:real_datasets}, we observe
that, with very few frequencies sampled, the RFF model with Tree sampling succeeded as well in
learning the distribution of the input datasets.
We conclude that this RFF method allows to closely approximate the trained quantum models.
We observe that for the two tasks, with a radically lower number of frequencies $(0.002 \times |\Omega| \simeq 160)$, the RFF accuracy/MSE loss performance mimics with fidelity the VQC performance even when more frequencies are used.
%The number of frequencies is radically lower since it needs only $0.002\times|\Omega| \simeq 160$. 

This result could indicate that the underlying distributions to learn were simple enough, such that the VQC had an excessive expressivity. 
In this case, and despite the large size of $\Omega$, the RFF manage to find similar solutions more quickly. 

%Regardless of the efficient learning of the underlying target function and considering the fact that task of optimizing the tunable parameters of the VQC is equivalent to the LRR using RFF (essentially with tree sampling), we can deduce that the classical model based on RFF suceeded in approximating the trained VQC function only by exploiting a very low dimensional kernel. 
%\helarq{ I am not really sure if this is the right conclusion: having same 'good' performance for RFF and VQC can be just a consequence of the 'easy' distribution of the dataset which can be learned with very basic models or does it prove an intrinsic relation between the trained VQC and the trained RFF??!}

\begin{figure*}
  \centering
  \subcaptionbox{Fashion-MNIST dataset (classification)}
  [.45\textwidth][c]{%
    \includegraphics[width=\linewidth]{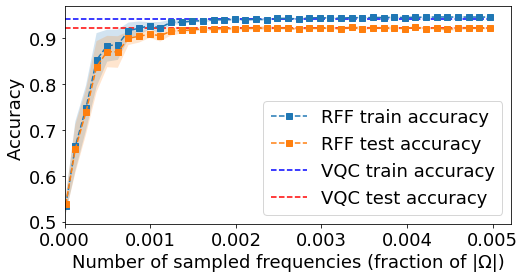}}
    \label{fig:mnist}
    \quad
  \subcaptionbox{California Housing dataset (regression)}
  [.45\textwidth][c]{%
    \includegraphics[width=\linewidth]{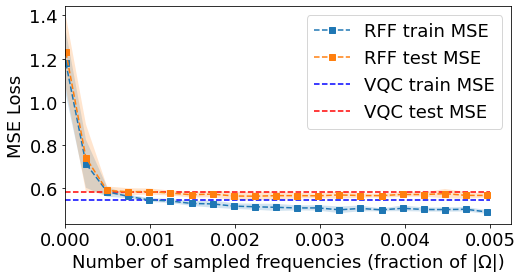}}
    \label{fig:housing}
\caption{\textbf{On real datasets.} Prediction results of a VQC and the its classical RFF approximator (Tree Sampling) on two 5-dimensional real datasets. A very low number of frequency samples is necessary to obtain similar results.}
\label{fig:real_datasets}   
\end{figure*}

\subsection{Number of Samples and Size of $\Omega$}
\label{sec:exp_RFF_scaling_with_Omega}

%In agreement with the theoretical bound, the number of samples $D$ given as a fraction of $|\Omega|$ decreases with the growth of the data input dimension and the number of encoding gates.
\begin{figure*}
  \centering
  \subcaptionbox{Experimental bounds for different values of $L, d$.}[.47\textwidth][c]{%
    \includegraphics[width=0.99\linewidth]{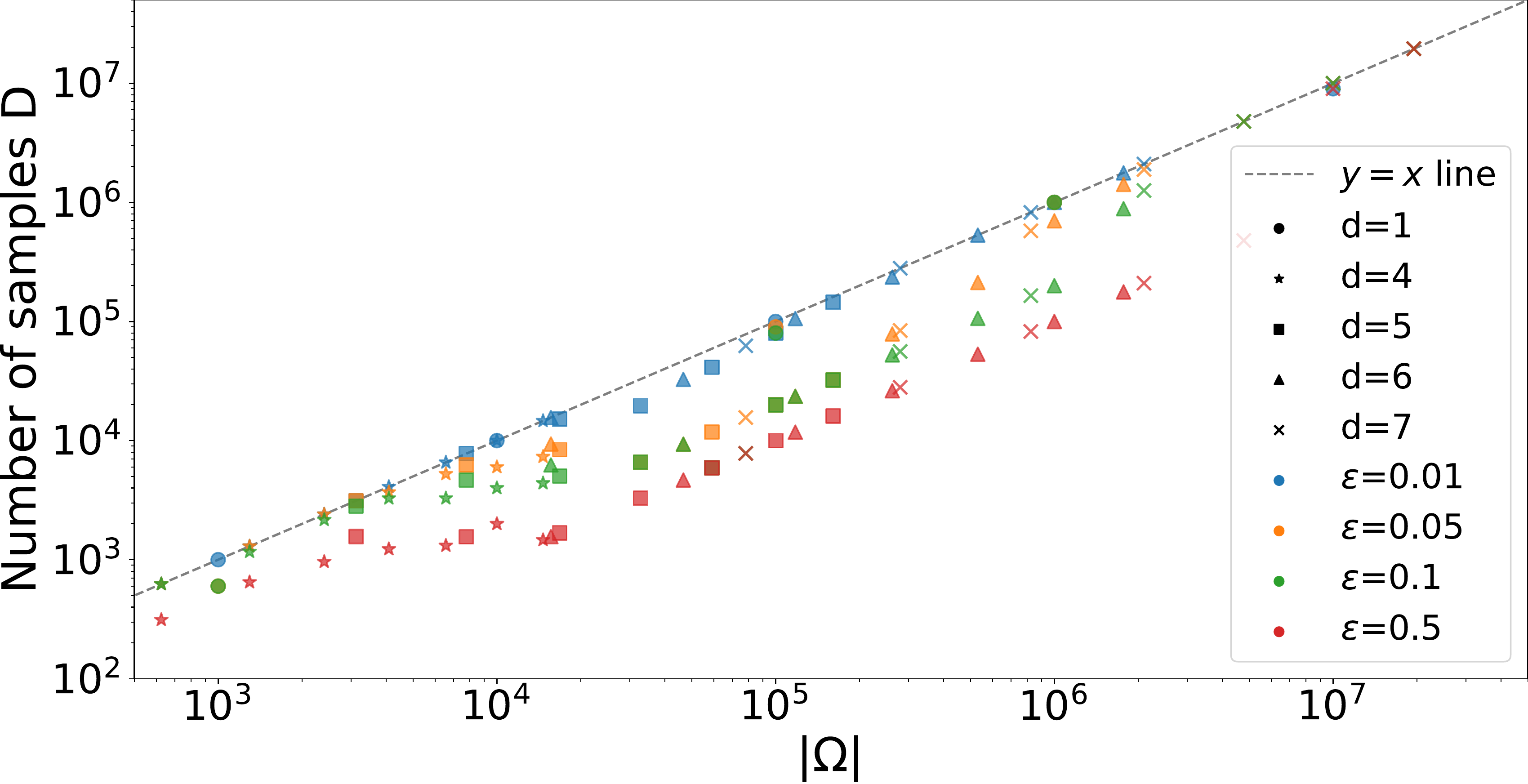}}
    \quad
  \subcaptionbox{Theoretical bound for different values of $L, d$.}[.47\textwidth][c]{%
    \includegraphics[width=\linewidth]{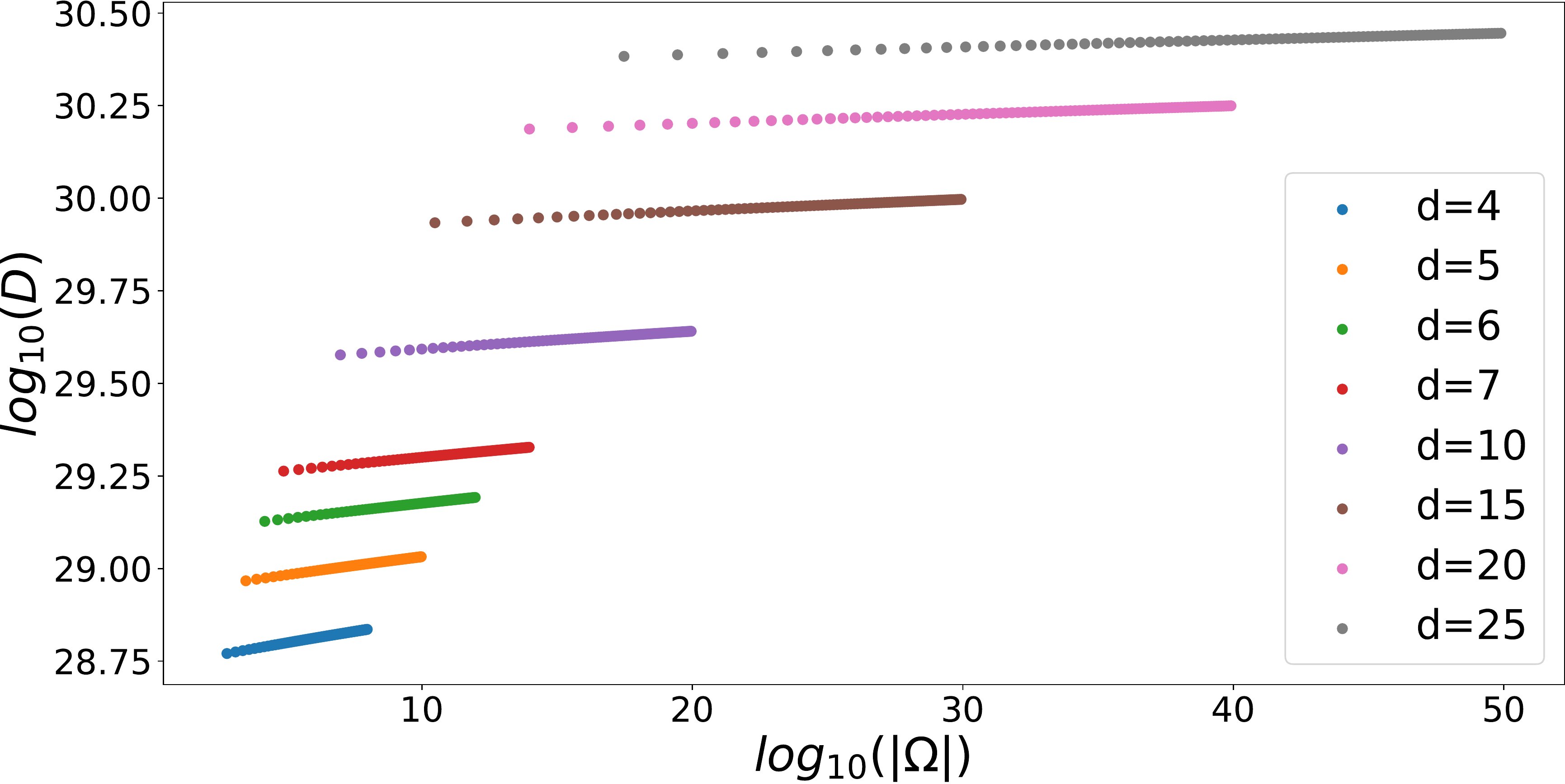}}
\caption{\textbf{Evolution of D as a function of input dimension d and of L encoding gates per dimension, and theoretical bounds.} In agreement with the theoretical bound, the number of samples $D$ given as a fraction of $|\Omega|$ decreases with the growth of the data input dimension and the number of encoding gates.}
\label{fig:impact_D_and_L}    
\end{figure*}

In this section, we test the theoretical bound provided by the theorem \ref{thm:number_of_samples_grid_sampling}. Given a spectrum $\Omega = \llbracket 0, L \rrbracket^d$, a Fourier series model trained on a specific dataset, the theorem bounds the necessary number of samples for a RFF model to approximate the original model with an $\epsilon$ error. This is an approximation to the Pauli encoding VQCs where the spectrum is $\Omega = \llbracket -L, L \rrbracket^d$.
For fixed values of $L, d$, and a spectrum $\Omega = \llbracket 0, L \rrbracket^d$, we implement the following protocol to test this bound:
\begin{itemize}
    \item Generate a dataset of $10^5$ points sampled uniformly from $[0, 1]^d$ and labels coming from a Fourier series on $\Omega$ with coefficients chosen uniformly from $[0, 1/\sqrt{|\Omega|}]$, split into a train set and a test set with respective fractions .9 and .1.
    \item For each value of D in $\{1, k|\Omega|/10 \text{ for }k\in \llbracket 1, 10\rrbracket\}$, sample $D$ frequencies from $\Omega$ without replacement, and train a linear ridge regression with $\lambda=10^{-6}$ on the train set. We performed the training with a Adam optimizer, a learning rate of .001, and between 50 and 200 epochs depending on the size of the spectrum with the constraint of computation time. Compute the output on the test set.
    \item Compute the mean absolute error between the output of the trained model with all the frequencies and the output of all other model. Select the model with the lowest number of samples that has an error below $\epsilon$.
    
\end{itemize}

The results of the application of this protocol are shown figure~\ref{fig:impact_D_and_L}. For the values of $|\Omega|$ between $10^4$ and $10^6$, one can see a significant reduction of the number of samples needed to approximate the whole model. For $\epsilon=.05$, one can expect to need only half of the spectrum, whereas for $\epsilon=.5$, one only need about 10\% of the spectrum. The trend does not continue above $|\Omega|=10^7$.

There are several limitations to this experiment. The main one is the limited training of the models. For the biggest values of $|\Omega|$ we limit ourselves to 50 epochs, which may be not enough to reach the optimal parameters, and thus blur the interpretation. Furthermore although the theorem is valid for every number of data points, the overparameterized regime where there are much more parameters than data points is known to exhibit unusual effects in linear regression \citep{hastie2022surprises}.

Given the choices of $\lambda$ and $\epsilon$, the theoretical bounds are very high for the regimes we experimentally tested, so they are not relevant. The effect that is quantified by the theory appears from $|\Omega|=10^{30}$, e.g one need approximately $D=10^{30}$ samples to approximate $10^{40}$ frequency which is still unfeasible on a classical computer.

\subsection{Discussion}

Overall, the numerical simulations were able to confirm a number of theoretical findings.

Regarding the expressivity of the VQCs, their spectrum is predictable from the encoding gates, as expected. However, we observed that the actual set of frequencies that emerge with non-zero coefficients don't cover the whole spectrum, questioning the effective expressivity of VQCs in practice. It confirms the intuition that frequencies that appear with high redundancy in the spectrum tend to have larger coefficients. 

The three RFF alternatives we proposed to approximate VQCs on several tasks revealed to be efficient. The number $D$ of samples (seen as a fraction of the spectrum size) grows favorably and allows for good approximation. 

In particular, when the redundancy of the spectrum is high, our \emph{Tree Sampling} method, which is the most computationally efficient, revealed to be able to approximate VQCs with fewer frequency samples. On the other hand, Tree Sampling inherits from the same drawbacks as VQCs and are sometimes not capable of learning less redundant frequencies in the VQCs spectrum, whereas \emph{Distinct} and \emph{Grid} sampling can outperform the VQCs in such cases. 

Even if these experiments are encouraging and match our expectations, we are far from the maximal efficiency of RFF methods. Indeed, the scaling of $D$, the number of samples, becomes even more favorable when the input dimension $d$ increase, or more generally when the spectrum has an exponential size. We expect to get closer to the theoretical bounds and see the number of samples becoming an even smaller fraction of the size of the spectrum. This was confirmed in our small scale simulations up to $d \leq 7$, limited by classical resources.
We have seen that, in order for the RFF to match some VQCs, in some cases the approximation requires using 80\% of the frequencies, in others it requires only 0.001\%.

We could go further by training VQCs on actual quantum computers and comparing them to classical RFF methods in practice.

% \begin{table*}[htbp]
%   \centering
%   \begin{tabular}{|l||c|c|c|c|c|}\hline
%     Characteristics    			& SVM (classical)	& VQC Pauli & $\text{RFF}_1$	& VQC Ham & $\text{RFF}_2$ \\\hline
%   %  Reddit$_{2000}$ 		& 				& $73.3\pm?$		&				& $73.3\pm?$ 	&		&  \\
%     accuracy	& NA		& NA 	& NA	& NA	& NA\\
%     GPU Time complexity & NA		& NA 	& NA	& NA	& NA\\\hline
%   \end{tabular}
%   \caption{\jonas{(Do we keep this??)} Comparison of a classical kernel, a variational quantum circuit (VQC) and the RFF approximator. We specifically look at the MNIST-Fashion dataset with $d=2$.}
%   \label{tab:summary}
% \end{table*}

\section{Conclusion}

In this work, we have studied the potential expressivity advantage of Variational Quantum Circuits (VQCs) for machine learning tasks, by providing novel classical methods to approximate VQCs. Our three methods use the fact that sampling few random frequencies from a large dimensional kernel (exponentially large in the case of VQCs) can be enough to provide a good approximation. This can be done given only the description of the VQCs and doesn't require running it on a Quantum Computer. We studied in depth the number of samples and its dependence to key aspects of the VQCs (input dimension, encoding Hamiltonians, circuit depth, VQC spectrum, number of training points). On the theoretical side, we conclude that our classical sampling method can approximate any VQC for machine learning tasks, with a number of sample that should scale favorably, but potentially high constant factors. Experimentally, we have tested our classical approximators on several use cases, using both artificial and real datasets, and our classical methods were able to match the VQC's results if not surpass them. 

These results question the definition of quantum advantage in the context of VQCs, shed some light on their effective expressivity, could pave the way for other classical approximators, but could also help to understand on what the true power of VQCs rely on.

In particular, it opens up questions about alternative encoding schemes, harnessing the effective expressivity of VQCs, and the link between a full expressivity and trainability.

\section{Acknowledgement}
This work was supported by the Engineering and Physical Sciences
Research Council (grants EP/T001062/1) and the H2020-FETOPEN Grant
PHOQUSING (GA no.: 899544).

\newpage

\bibliographystyle{IEEEtran}
\bibliography{refs}

\appendix

\section{Definitions of Linear Ridge Regression and Kernel Ridge Regression}
\label{app:lrr}
We present in this section the Linear Ridge Regression (LRR) and Kernel Ridge Regression (KRR) problem \cite{bishop2006pattern}. The problem of regression is to predict continuous label values from feature vectors. We are given a dataset $\{(x_i, y_i), i\in \llbracket 1, M \rrbracket x_i \in \mathbb{R}^d, y_i \in \mathbb{R} \}$, and to each data point $x$ an associated feature vector $\phi(x) \in \mathbb{R}^p$. The goal of LRR is to construct a parameterized model $f$ such that $f(x) = y$. The model is parameterized by a weight vector $\w$ of size $p$ such that $f(x; w) = \w^T\phi(x)$. Training the model consists of finding the vector $\w*$ that minimizes the loss function
\begin{gather}
    \w^* = \argmin_{\w} \frac{1}{M}\Sum_{i=1}^M|\w^T\phi(x_i) - y_i|^2 + \lambda||\w||^2\\
    = \argmin_{\w} \frac{1}{M}||\boldsymbol{\Phi}\w - \boldsymbol{y}||^2 + \lambda||\w||^2\\
\end{gather}
where $\boldsymbol{\Phi}$ is a matrix of size $M\times p$ with each row $i$ corresponds to $\phi(x_i)^T$ and $\boldsymbol{y}$ is the vector of all the labels $y_i$. 
The first term of the loss is the Mean Square Error (MSE) and corresponds to the difference between the prediction and the ground truth. The second term is the ridge regularization, and prevents the weights from exploding. The magnitude of the regularization is controlled by the hyperparameter $\lambda>0$.

When $p < M$, an analytic solution to this problem is given by $\w^* = (\boldsymbol{\Phi}^T\boldsymbol{\Phi} + M\lambda I_p)^{-1}\boldsymbol{\Phi}^T\boldsymbol{y}$.

As a consequence, to make the LRR possible and have a single solution, the number of training points must be larger than the number of features in the feature space ($\phi(x)$). Otherwise, one can perform a gradient descent. 

The dual formulation of this problem is given by expressing $\w$ as a linear combination of the data points $\w = \boldsymbol{\Phi}^T\alphab$. The minimization on $\w$ become a minimization on $\alphab$ and can be expressed as 
\begin{gather}
    \alphab^* = \argmin_{\alphab}  \frac{1}{M}||\Phib\Phib^T\alphab - \boldsymbol{y}||^2 + \lambda \alphab^T\Phib\Phib^T\alphab\\
\end{gather}

The solution of this problem is $\alphab = (\Phib\Phib^T + M\lambda I_M)^{-1} \boldsymbol{y}$.

Note that the dual solution only depends on the matrix of scalar products between feature vectors $\Phib\Phib^T$. One can then replace this matrix by a kernel matrix $K$ and the obtained model is a Kernel Ridge Regression.

\section{Approximation results for Random Fourier Features}
\label{app:rff_bounds}
%\subsection{Analytical Bounds}

We give here two useful results about the bounds of the error of the RFF method. RFFs are supposed to approximate a certain kernel $k$ by using fewer features.  Intuitively, not enough features would lead to imprecise solutions. The following theorems \cite{Rahimi2009, sutherland2015error} bounds the error obtained when comparing the kernel $k(x, y)$ by the RFF approximator $\phi(x)^T\phi(y)$ using $D$ samples.

We recall that the condition on the kernel $k$ is for it to be expressed as
\begin{gather}
    k(\delta) = \Int_{\omega \in \mathcal{X}} p(\omega) e^{-i\omega^T\delta}d\omega
\end{gather}
where $p(\omega)$ is the distribution of the frequencies $\omega$. 
\begin{theorem}
\label{thm:approx_kernel}
Let $\mathcal{X}$ be a compact set of $\mathbb{R}^d$, and $\epsilon > 0$. 
\begin{equation}
\begin{split}
    \mathbb{P}(\sup_{x, y \in \mathcal{X}}|k(x-y) - \phi(x)^T\phi(y)|\geq \epsilon) \leq \\
    66(\frac{\sigma_p |\mathcal{X}|}{\epsilon})^2 exp(-\frac{D\epsilon^2}{4(d+2)})
\end{split}
\end{equation}
with $\sigma_p^2 = \mathbb{E}_p(\omega^T\omega)$, the variance of the frequencies' distribution, and $|\mathcal{X}| = max_{x,x'\in\X}(\|x-x'\|)$ the diameter of $\mathcal{X}$.
\end{theorem}

The following theorem \cite{sutherland2015error} bounds the actual prediction error when using RFF compared to the KRR estimate. The formula in the original reference contains a sign error and we correct it here.
\begin{theorem}
\label{thm:approx_krr}
Let $\mathcal{X}$ be a compact set of $\mathbb{R}^d$, and $\epsilon > 0$. We consider a training set $\mathcal{D}\{(x_i, y_i)\}_{i=1}^M$. Let $f$ be the KRR model obtained with the true kernel $k$ and regularization $\lambda = M\lambda_0$ for $\lambda_0>0$, and $\tilde{f}$ the KRR model obtained with the approximate kernel and the same regularization. Then we can guarantee $|f(x) - \tilde{f}(x)| \leq \epsilon$ with probability $1 - \delta$ for a number $D$ of samples given by: 
\begin{equation}\label{eq:theorem_samples_vs_error}
\begin{split}
    D = \Omega\Bigg(d\bigg(\frac{(\lambda_0+1)\sigma_y}{\lambda_0^2\epsilon}\bigg)^2\bigg[log(\sigma_p|\mathcal{X}|)\\
    + log\frac{(\lambda_0+1)\sigma_y}{\lambda_0^2\epsilon} - log\delta\bigg]\Bigg)
\end{split}
\end{equation}
with $\sigma_y^2 = \frac{1}{M}\sum_{i=1}^My_i^2$ and $\sigma_p$, $|\mathcal{X}|$ being defined in theorem \ref{thm:approx_kernel}. We recall that in Eq.\ref{eq:theorem_samples_vs_error} the notation $\Omega$ stands for the computational complexity "Big-$\Omega$" notation. 
\end{theorem}

%The Fourier Feature method then allows us to get an approximation of the underlying kernels of quantum models with an error bounded from the above theorem. We will look into the details of the constants of this theorem in the different cases of quantum circuits. 

%\section{Bounds for every sampling strategy}
\section{Approximation results for RFF in the context of VQCs}

\subsection{Distinct sampling in the Pauli encoding case}
\label{app:proof_pauli}
In the case of Pauli encoding only, we know that $\Omega = \llbracket-L,L\rrbracket^d$ (considered here to be the full spectrum, not $\Omega_+$ defined in Section \ref{sec:quantum_models_are_shifinvariantkernel}, which would have been equivalent). In one dimension, we simply have $\sigma_p = 1/L\sum_{\ell=-L,\cdots,L}\ell^2 = O(L^2)$. In dimension $d$, a frequency $\omega$ is given by its values on each dimension $(j_1,\cdots,j_d)$ with $j_k\in[|-L,L|]$. We similarly have
\begin{equation}
    \sigma_p = \frac{1}{(2L+1)^d}\sum_{j_1,\cdots,j_d} j_1^2+\cdots+j_d^2 
\end{equation}

%\jonas{Constantin: "Isn't the variance of a random vector defined as the matrix containing the individual variances and covariances?"}
Note that $\sum_{j_1,\cdots,j_d}j_1^2+\cdots+j_d^2$ is $d(2L+1)^{d-1}$ times the sum of all squares, 
\begin{equation}
\begin{split}
    \sigma_p = \frac{d(2L+1)^{d-1}}{(2L+1)^d} \sum_{\ell=-L}^{L}\ell^2
    = \frac{d}{2L+1}\frac{2L(L+1)(2L+1)}{6} \\
    = O(dL^2) = O(d|\Omega|^{2/d})
\end{split}    
\end{equation}

The expression is then obtained by replacing the value of $\sigma_p$ in theorem \ref{thm:approx_krr}.

We note that we can generalize this results to scaled Pauli encoding, as done in \cite{shin2022exponential}, by replacing $L$ by a term growing as $c^L$ where $c$ is a constant. $D$ would grow linearly in $L$ and not logarithmically anymore.

\subsection{Grid sampling with a general hamiltonian}
\label{sec:gridsampling_proof}

The following theorem  bounds the approximation between a function defined by its Fourier series and another function with frequencies distant by at most a constant $s$ of the original frequencies.

Let $\mathcal{X}$ a compact set of $\mathbb{R}^d$ with diameter $|\mathcal{X}|$ and $\Omega$ a finite subset of $\mathcal{X}$. Let $f(x) = \Sum_{\omega \in \Omega} a_\omega cos(\omega^Tx) + b_\omega sin(\omega^Tx)$. Let $\Omega'$ a subset of $\mathcal{X}$ and $s>0$ such that $\forall \omega \in \Omega, \:\exists \omega'\in\Omega, \: |\omega-\omega'|\leq s$.

%For each $\omega$ let $\b(\omega)$ be such an element of $\Omega'$. There is no necessarily uniqueness in the choice of $\b(\omega)$.

Let $\mathcal{F}_{\Omega'} = \bigg\{ \Sum_{\omega \in \Omega'} a_\omega cos(\omega^Tx) + b_\omega sin(\omega^Tx), a_\omega, b\omega \in \mathbb{R} \bigg\}$.

\begin{theorem}
It exists f' $\in \mathcal{F}_{\Omega'}$ such that 
\begin{equation}
    \sup_{x \in \mathcal{X}}|f'(x) - f(x)| \leq s C
\end{equation}
with $C = |\mathcal{X}| |f|_\infty$.
\end{theorem}

\begin{proof}
For each $\omega \in \Omega$ let $b(\omega) \in \Omega'$ be such that $|\omega - b(\omega)|\leq s$. Such element exists by definition but is not necessarily unique. Let $f'(x) = \Sum_{\omega \in \Omega} a_\omega cos(b(\omega)^Tx) + b_\omega sin(b(\omega)^Tx)$. The $b(\omega)$s are not necessarily different therefore there might be less frequencies in $f'$ than in $f$. 

\begin{equation}
\begin{split}
    |f(x) - f'(x)| = 2\bigg|\Sum_{\omega \in \Omega} sin(\frac{(\omega - b(\omega))^T}{2}x)\\
    [b_\omega sin(\frac{(\omega + b(\omega))^T}{2}x) -
    a_\omega cos(\frac{(\omega + b(\omega))^T}{2}x)] \bigg|\\
                \leq 2\Sum_{\omega \in \Omega}\big| \frac{(\omega - b(\omega))^T}{2}\big||x|[|b_\omega|+ |a_\omega|] \\
                \leq s |x|\Sum_{\omega \in \Omega}|b_\omega|+ |a_\omega| \\
                \leq s |\mathcal{X}| |f|_\infty
\end{split}
\end{equation}

\end{proof}

We shall here extend the proof where we sample from the grid described above. Let us note $\hat{f_s}$ the RFF model with the whole grid and $\tilde{f}$ the RFF model with D samples from the grid below. For all $x \in \mathcal{X}$ we have 
\begin{equation}
    |\tilde{f}(x) - f(x)| \leq |\tilde{f}(x) - \hat{f_s}| + |\hat{f_s} - f(x)|\\
    \leq |\tilde{f}(x) - \hat{f_s}| + sC
\end{equation}
Then 
\begin{equation}
    \PP(|\tilde{f}(x) - f(x)| \geq \epsilon) \leq \PP(|\tilde{f}(x) - \hat{f_s}| \geq \epsilon - sC)
\end{equation}

for $s < \epsilon/C$.

In this case $|\Omega| = (\omega_{max}/s)^d$.
Using the expression of Section \ref{app:proof_pauli}, we can guarantee that $|f(x) - \tilde{f}(x)| \leq \epsilon$ with probability $1 - \delta$ if 
\begin{equation}
    D = \Omega\Bigg(d\frac{1}{(\epsilon - sC)^2}\bigg[log(\omega_{max}/s) + log\frac{1}{\epsilon - sC} - log\delta\bigg]\Bigg)
\end{equation}

\end{document}